\documentclass[draftcls,onecolumn]{IEEEtran}
\IEEEoverridecommandlockouts

\usepackage{amsmath}
\usepackage{amssymb}
\usepackage{amsfonts}
\usepackage{float} 
\usepackage{srcltx}
\usepackage{mathrsfs} 
\usepackage{multirow}

\usepackage{graphicx}
\usepackage{epsfig}
\usepackage{psfrag}
\usepackage{subfigure}

\usepackage{setspace}
\usepackage{dsfont} 

\usepackage[]{units} 
\usepackage{url} 
\usepackage[dvips]{color}
\usepackage{verbatim} 
\usepackage[numbers,square,comma,compress]{natbib} 
\usepackage{ifthen} 
\usepackage{ifpdf} 

\usepackage{ulem} 
\usepackage{arydshln} 
\usepackage{stackrel} 
\usepackage{mathtools} 

\newcommand{\NL}[3]{\ensuremath{ #1 \nonumber\\&\phantom{#2} #3 }} 
\newcommand{\NLa}[1]{\ensuremath{ \NL{\right.}{#1}{\left.} }} 

\newcommand{\card}[1]{\ensuremath{\left\lvert#1\right\rvert}}

\newcommand{\pr}[1]{\ensuremath{\mathbb{P}\left( #1 \right)}}
\newcommand{\indicator}[1]{\ensuremath{\mathds{1}_{\left\{ #1 \right\}}}}
\newcommand{\ceil}[1]{\ensuremath{\left\lceil#1\right\rceil}}

\newcommand{\markov}{\ensuremath{\leftrightarrow}}

\newcommand{\pa}[1]{\ensuremath{ \left( #1 \right) }}
\newcommand{\paa}[1]{\ensuremath{ \mathopen{}\left( #1 \right)\mathclose{} }}
\newcommand{\pb}[1]{\ensuremath{ \left[ #1 \right] }}
\newcommand{\pc}[1]{\ensuremath{ \left\{ #1 \right\} }}

\newcommand{\BF}{\ensuremath{ \{0,1\} }}

\newcommand{\mspca}{\ensuremath{\,}}

\newcommand{\mspcc}{\ensuremath{\;}}
\newcommand{\mspcd}{\ensuremath{\quad}}

\newcommand{\mset}[1]{\ensuremath{\mathcal{#1}}}
\newcommand{\mvec}[1]{\ensuremath{\textbf{#1}}}

\newcommand{\mDefine}{\ensuremath{ \stackrel{\triangle}{=} }}

\newcommand{\dotleq}{\ensuremath{ \stackrel{\centerdot}{\leq} }}

\newcommand{\mGoesTo}{\ensuremath{\rightarrow}}
\newcommand{\mGoesToAs}[1]{\ensuremath{\underset{#1}{\longrightarrow}}}

\newcommand{\mRR}{\ensuremath{\mathbb{R}}}

\newcommand{\mZZ}{\ensuremath{\mathbb{Z}}}
\newcommand{\mNN}{\ensuremath{\mathbb{N}}}

\newcommand{\mDistributedBy}{\ensuremath{\sim}}

\newcommand{\D}[2]{\ensuremath{ D\pa{{#1} \middle\| {#2}} }}
\newcommand{\Dcond}[3]{\ensuremath{ D\pa{{#1} \| {#2} \mid {#3}} }}

\renewcommand{\H}[1]{\ensuremath{ H\left( {#1} \right) }}
\newcommand{\Hcond}[2]{\ensuremath{ H\left( {#1} \middle| {#2} \right) }}

\newcommand{\ckI}[2]{\ensuremath{I\left({#1},{#2}\right)}}
\newcommand{\ckIcond}[3]{\ensuremath{I\left({#1},{#2}\mid{#3}\right)}}

\newcommand{\Hamming}{\ensuremath{ \mathtt{H} }} 
\newcommand{\dB}[2]{\ensuremath{ D_b\paa{{#1} \middle\| {#2}} }} 
\newcommand{\hB}[1]{\ensuremath{ H_b\paa{#1} }} 
\newcommand{\hBinv}[1]{\ensuremath{ {H_b^{-1}\left({#1}\right)} }} 
\newcommand{\wH}[1]{\ensuremath{ w_\Hamming{}\paa{#1} }} 
\newcommand{\nwH}[1]{\ensuremath{ \delta_\Hamming\paa{#1} }} 

\newcommand{\Ber}[1]{\ensuremath{ \mDistributedBy \text{Ber} \left( #1 \right) }}
\newcommand{\BerV}[2]{\ensuremath{ \mDistributedBy \text{BerV} \left( {#1}, {#2} \right) }}
\newcommand{\Binomial}[2]{\ensuremath{ \mDistributedBy \text{Binomial} \left( {#1},{#2} \right) }}
\newcommand{\Uniform}[1]{\ensuremath{ \mDistributedBy \text{Uniform} \left( #1 \right) }}

\newcommand{\deltamin}{\ensuremath{ {\delta_\text{min}} }}
\newcommand{\deltaGV}{\ensuremath{ {\delta_\text{GV}} }}

\newcommand{\cc}{\ensuremath{\mvec{c}}}

\renewcommand{\SS}{\ensuremath{\mvec{S}}}
\renewcommand{\ss}{\ensuremath{\mvec{s}}}
\newcommand{\UU}{\ensuremath{\mvec{U}}}
\newcommand{\uu}{\ensuremath{\mvec{u}}}

\newcommand{\vv}{\ensuremath{\mvec{v}}}
\newcommand{\XX}{\ensuremath{\mvec{X}}}
\newcommand{\xx}{\ensuremath{\mvec{x}}}
\newcommand{\YY}{\ensuremath{\mvec{Y}}}
\newcommand{\yy}{\ensuremath{\mvec{y}}}
\newcommand{\NN}{\ensuremath{\mvec{N}}}

\newcommand{\zz}{\ensuremath{\mvec{z}}}
\newcommand{\ZZ}{\ensuremath{\mvec{Z}}}

\newcommand{\tY}{\ensuremath{{\tilde{Y}}}}

\newcommand{\sC}{\ensuremath{\mset{C}}}
\newcommand{\sE}{\ensuremath{\mset{E}}}

\newcommand{\sT}{\ensuremath{\mset{T}}}
\newcommand{\sX}{\ensuremath{\mset{X}}}
\newcommand{\sY}{\ensuremath{\mset{Y}}}
\newcommand{\sZ}{\ensuremath{\mset{Z}}}

\newcommand{\Csiszar}{Csisz\'ar}

\newcommand{\KL}{Kullback-Leibler}
\newcommand{\KM}{K\"orner-Marton}

\newcommand{\GV}{Gilbert-Varshamov}

\renewcommand{\mid}{\middle\vert}
\newcommand{\mcomp}[1]{\ensuremath{\overline{#1}}} 

\newcommand{\thd}{\ensuremath{ t }} 
\newcommand{\Cn}{\ensuremath{ \sC^{(n)} }} 
\newcommand{\OmegaO}{\ensuremath{ \Omega_\mvec{0} }}
\newcommand{\SequenceOfCodes}{\ensuremath{ \Cn\subseteq\BF^n,\text{ }n=1,2,\ldots }} 

\newcommand{\BestKnownAchievableErrExp}[2]{\ensuremath{{ \underline{E}_{\text{BSC}}\left({#1},{#2}\right) }}}

\newcommand{\tradeoffE}[4]{\ensuremath{{ \underline{E}_{{#1}}^{\text{({#2})}} \left({#3}; {#4}\right) }}}

\newcommand{\swEo}[2]{\tradeoffE{0}{SI}{#1}{#2}}
\newcommand{\swEl}[2]{\tradeoffE{1}{SI}{#1}{#2}}
\newcommand{\kmEo}[2]{\tradeoffE{0}{KM}{#1}{#2}}
\newcommand{\kmEl}[2]{\tradeoffE{1}{KM}{#1}{#2}}

\newcommand{\decider}{{decision function}}

\newcommand{\steinE}{\ensuremath{{ \sigma }}}
\newcommand{\steinEoX}{\ensuremath{{ \sigma_{X,0} }}}
\newcommand{\steinElX}{\ensuremath{{ \sigma_{X,1} }}}
\newcommand{\steinEoSym}{\ensuremath{{ \sigma_0 }}}
\newcommand{\steinElSym}{\ensuremath{{ \sigma_1 }}}

\newcommand{\Type}[2]{\ensuremath{{ \sT_{#1}\left({#2}\right) }}} 

\newcommand{\Ball}[3]{\ensuremath{{ \mset{B}_{#1}\left({#2},{#3}\right) }}} 

\newcommand{\UCBall}[2]{\ensuremath{{ \Uniform{\Ball{#1}{\mvec{0}}{#2}} }}} 

\newcommand{\EBTinBallSymbol}{\ensuremath{ E_\text{BT} }}
\newcommand{\EBTinBall}[4]{\ensuremath{ \EBTinBallSymbol\left( {#1}, {#2}, {#3}, {#4} \right)}} 
\newcommand{\EBBinball}[4]{\ensuremath{ E_\text{BB}\left( {#1}, {#2}, {#3}, {#4} \right)}} 

\newcommand{\DhtAchv}{\ensuremath{ \mathscr{C} }} 
\newcommand{\DhtAchvX}{\ensuremath{ \mathscr{C}_X }} 
\newcommand{\DhtAchvSym}{\ensuremath{ \mathscr{C} }} 

\newcommand{\HoMarker}{\ensuremath{ (0) }} 
\newcommand{\HlMarker}{\ensuremath{ (1) }} 

\newcommand{\HMark}[2]{\ensuremath{ {#1}^{#2} }} 

\newcommand{\HypVar}[2]{\ensuremath{ \HMark{#1}{#2} }} 

\newcommand{\PH}[2]{\ensuremath{ {\HMark{P}{#1}_{#2}} }} 
\newcommand{\Po}[1]{\ensuremath{ \PH{\HoMarker}{#1} }}
\newcommand{\Pl}[1]{\ensuremath{ \PH{\HlMarker}{#1} }}
\newcommand{\Pt}[1]{\ensuremath{ \PH{(*)}{#1} }}

\newcommand{\dotD}[1]{\ensuremath{ \mspca \stackbin[{\scriptstyle D}]{\centerdot}{{#1}} \mspca }}
\newcommand{\doteqD}{\ensuremath{{ \dotD{=} }}}
\newcommand{\dotleqD}{\ensuremath{{ \dotD{\leqslant} }}}

\newcommand{\mat}[1]{\ensuremath{ \text{#1} }} 

\newcommand{\Rbin}{\ensuremath{{ R_\text{bin} }}}

\newcommand{\NrCoveringRadius}{\ensuremath{ \rho_\text{cover} }}
\newcommand{\NrPackingRadius}{\ensuremath{ \rho_\text{pack} }}

\newcommand{\dspectrum}[2]{\ensuremath{ \Gamma_{#1}(#2) }} 
\newcommand{\dspectrumgood}[3]{\ensuremath{ \underline{\Gamma}^{(#1)}_{#2}(#3) }} 

\newtheorem{lemma}{Lemma}
\newtheorem{corollary}{Corollary}
\newtheorem{theorem}{Theorem}
\newtheorem{definition}{Definition}

\newtheorem{proposition}{Proposition}

\newtheorem{remark}{Remark}

\begin{document}

\allowdisplaybreaks

\title{On Binary Distributed Hypothesis Testing}

\author{Eli Haim and Yuval Kochman}


\maketitle



\begin{abstract}
We consider the problem of distributed binary hypothesis testing of two sequences that are generated by an i.i.d.  doubly-binary symmetric source. Each sequence is observed by a different terminal. The two hypotheses correspond to different levels of correlation between the two source components, i.e., the crossover probability between the two. The terminals communicate with a \decider{} via rate-limited noiseless links. We analyze the tradeoff between the exponential decay of the two error probabilities associated with the hypothesis test and the communication rates. We first consider the side-information setting where one encoder is allowed to send the full sequence. For this setting, previous work exploits the fact that a decoding error of the source does not necessarily lead to an erroneous decision upon the hypothesis. We  provide improved achievability results by carrying out a tighter analysis of the effect of binning 
error;
the results are also more complete as they cover the full exponent tradeoff and all possible correlations. We then turn to the setting of symmetric rates for which we utilize \KM{} coding to generalize the results, with little degradation with respect to the performance with a one-sided constraint (side-information setting).
\end{abstract}

%
\section{Introduction}
%
\label{sec:intro}

We consider the distributed hypothesis testing (DHT) problem, where there are two distributed sources, $X$ and $Y$, and the hypotheses are given by
\begin{subequations}
\label{eq:hypothses:general}
\begin{align}
&\mathcal{H}_0: (X,Y) \mDistributedBy \Po{X,Y}
\\
&\mathcal{H}_1: (X,Y) \mDistributedBy \Pl{X,Y},
\end{align}
\end{subequations}
where $\Po{X,Y}$ and $\Pl{X,Y}$ are different joint distributions of $X$ and $Y$.
The test is performed based on information sent from two distributed terminals (over noiseless links), each observing $n$ i.i.d. realizations of a different source, where the rate of the information sent from each terminal is constrained. This setup, introduced in~\cite{Berger1979Spring,AhlswedeCsiszar1981}, introduces a tradeoff between the information rates and the probabilities of the two types of error events.
In this work we focus on the exponents of these error probabilities, with respect to the number of observations $n$.


When at least one of the marginal distributions depends on the hypothesis,
a test can be constructed based only on the type of the corresponding sequence. Although this test may not be optimal, it results in non-trivial performance (positive error exponents) with zero rate.
In contrast, when the marginal distributions are the same under both hypotheses, a positive exponent cannot be achieved using a zero-rate scheme, see~\cite{ShalabyP1992}.


One may achieve positive exponents while maintaining low rates, by effectively compressing the sources and then basing the decision upon their compressed versions. Indeed, many of the works that have considered the distributed hypothesis testing problem bear close relation to the distributed compression problem.

Ahlswede and \Csiszar{}~\cite{AhlswedeCsiszar1986} have suggested a scheme based on compression without taking advantage of the correlation between the sources; Han~\cite{Han1987} proposed an improved scheme along the same lines. 
Correlation between the sources is exploited by Shimokawa et al.~\cite{ShimokawaHanAmari1994ISIT, Shimokawa:1994:MScThesis} to further reduce the coding rate, incorporating random binning following the Slepian-Wolf~\cite{SlepianWolf73} and Wyner-Ziv~\cite{WynerZiv76} schemes. Rahman and Wagner~\cite{RahmanWagner2012} generalized this setting and also derived an outer bound. They also give a ``quantize and bin'' interpretation to the results of~\cite{ShimokawaHanAmari1994ISIT}.
Other related works include~\cite{HanKobayashi1989, Amari2011, Polyanskiy2012ISIT, Katz2017, Katz2016Asilomar}.
See~\cite{HanAmari1998, RahmanWagner2012} for further references.

We note that in spite of considerable efforts over the years, the problem remains open. In many cases, the gap between the achievability results and the few known outer bounds is still large. Specifically, some of the stronger results are specific to testing against independence (i.e., under one of the hypotheses $X$ and $Y$ are independent), or specific to the case where one of the error exponents is zero (``Stein's-Lemma'' setting). The present work significantly goes beyond previous works, extending and improving the achievability bounds.   Nonetheless, the refined analysis comes at a price. Namely, in order to facilitate analysis, we choose to restrict  attention to a simple source model.

To that end, we consider the case where $(X,Y)$ is a doubly symmetric binary source (DSBS). That is, $X$ and $Y$ are each binary and symmetric. Let $Z \mDefine Y \ominus X$ be the modulo-two difference between the sources.\footnote{Notice that in this binary case, the uniform marginals mean that $Z$ is necessarily independent of $X$.}
We consider the following two hypotheses:
\begin{subequations}
\label{eq:hypothses}
\begin{align}
&\mathcal{H}_0: Z \Ber{p_0}
\\
&\mathcal{H}_1: Z \Ber{p_1},
\end{align}
\end{subequations}
where we assume throughout that $p_0 \leq p_1 \leq 1/2$.
Note that a sufficient statistic for hypothesis testing in this case is the weight (which is equivalent to the type) of the noise sequence $\ZZ$. Under communication rate constraints, a plausible approach would be to use a distributed compression scheme that allows lossy reconstruction of the sequence $Z$, and then base the decision upon that sequence.

We first consider a one-sided rate constraint. That is, the $Y$-encoder is allocated the full rate of one bit per source sample, so that the $\YY$ sequence is available as side information at the decision function. In this case, compression of $\ZZ$ amounts to compression of $\XX$; a random binning scheme is optimal for this task of compression, lossless or lossy.\footnote{More precisely, it gives the optimal coding rates, as well as the best known error exponents when the rate is not too high.} Indeed, in this case, the best known achievability result is due to~\cite{ShimokawaHanAmari1994ISIT}, which basically employs a random binning scheme.\footnote{Interestingly, when $p_1=1/2$ (testing against independence), the simple scheme  of~\cite{AhlswedeCsiszar1986} which ignores the side-information altogether is optimal.} 

A natural question that arises when 
using binning as part of the distributed 
hypothesis testing scheme is the effect 
of a ``bin decoding error" on 
the decision error between the hypotheses. 
The connection between these two errors 
is non-trivial as 
a bin decoding error  inherently results in a ``large'' noise reconstruction error, much in common with errors in channel coding (in the context of syndrome decoding).
Specifically, when a binning error occurs, the reconstruction of the noise sequence $\ZZ$ is roughly consistent with an i.i.d. Bernoulli $1/2$ distribution. Thus, if one feeds the weight of this reconstructed sequence to a simple threshold test, it would typically result in deciding that the noise was distributed according to $p_1$, regardless of whether that is the true distribution or not.
This effect causes an asymmetry between the two error probabilities associated with the 
hypothesis test.
Indeed, as the Stein exponent corresponds to highly asymmetric error probabilities,  the exponent derived in~\cite{ShimokawaHanAmari1994ISIT} may be interpreted as taking advantage of this effect.\footnote{Another interesting direction, not pursued in this work, is to 
change the problem formulation to allow declaring an ``erasure" when 
the probability of a bin decoding error 
exceeds a certain threshold.}

The contribution of the present work is twofold. First we extend and strengthen 
the results of \cite{ShimokawaHanAmari1994ISIT}.
By explicitly considering and leveraging the properties of good codes, we bound the probability that the sequence $\ZZ$ happens to be such that $\YY \ominus \ZZ$ is very close to some wrong yet ``legitimate'' $\XX$, much like an undetected error event in erasure decoding~\cite{Forney68}.
This allows us to derive achievability results for the full tradeoff region, namely the tradeoff between the error exponents corresponding to the two types of hypothesis testing errors. 

The second contribution is in considering a symmetric-rate constraint. For this case, the optimal distributed compression scheme for $Z$ is the \KM{} scheme~\cite{KornerMarton79}, which requires each of the users to communicate at a rate $\H{Z}$; hence, the sum-rate is strictly smaller than the one of Slepian-Wolf, unless $Z$ is symmetric. Thus, the \KM{} scheme is a natural candidate for this setting.
Indeed, it was observed in~\cite{AhlswedeCsiszar1986, HanAmari1998} that a standard information-theoretic solution such as Slepian-Wolf coding may not always be the way to go, and  \cite{HanAmari1998} mentions the the \KM{} scheme in this respect.
Further, Shimokawa and Amari~\cite{ShimokawaAmari:ISIT:1995} point out the possible application of the \KM{} scheme to distributed parameter estimation in a similar setting and  a similar observation is made in~\cite{GamalL15AllertonArxiv}.
However, to the best of our knowledge, the present work is the first to propose an actual \KM{}-based scheme for distributed hypothesis testing and to analyze its performance. Notably, the performance tradeoff obtained recovers the achievable tradeoff derived for a one-sided constraint.

The rest of this paper is organized as follows. In Section~\ref{sec:problem_statement_notations} we formally state the problem, define notations and present some basic results. Section ~\ref{sec:related_results} 
and \ref{sec:linear_codes_and_EE} provide necessary background: the first surveys known results for the case of a one-sided rate constraint while the latter provides definitions and properties of good linear codes. 
In Section~\ref{sec:one_user_constrained_case} we present the derivation of a new achievable exponents tradeoff region. Then, in Section~\ref{sec:symmetric_rate_constraint} we present our results for a symmetric-rate constraint. 
Numerical results and comparisons appear in Section~\ref{sec:performance_comparison}.
Finally, Section~\ref{sec:future_work}  concludes the paper.

%
\section{Problem Statement and Notations}
%
\label{sec:problem_statement_notations}

\subsection{Problem Statement}
\label{sec:problem_statement}

\newcommand{\tpxcaptionproblemsetup}{Problem setup.}
\newcommand{\tpxlabelproblemsetup}{fig:problem_setup}
\input{problem_setup.tpx}

Consider the setup depicted in Figure~\ref{fig:problem_setup}. $\XX$ and $\YY$ are random vectors of blocklength $n$, drawn from the (finite) source alphabets $\sX$ and $\sY$, respectively. 
Recalling the hypothesis testing problem~\eqref{eq:hypothses:general}, we have two possible i.i.d. distributions. 
In the sequel we will take a less standard notational approach, and define the hypotheses by random variable $H$ which takes the values $0,1$, and assume a probability distribution function $P_{X,Y|H}$; Therefore $H=i$ refers to $\mathcal{H}_i$ of~\eqref{eq:hypothses:general} and~\eqref{eq:hypothses}.\footnote{We do not assume any given distribution over $H$, as we are always interested in probabilities given the hypotheses.} We still use for the distribution $P_{X,Y|H=i}$ (for $i=0,1$) the shortened notation $\PH{(i)}{X,Y}$. 
Namely, for any $\xx \in \sX^n$ and $\yy \in \sY^n$, and for $i \in \{0,1\}$,
\begin{align*}
\pr{ \XX=\xx,\YY=\yy|H=i } = \prod_{j=1}^n \PH{(i)}{X,Y}(x_j,y_j).
\end{align*}
A scheme for the problem is defined as follows.


\newcommand{\scheme}{\ensuremath{\Upsilon}}
\begin{definition}
A scheme $\scheme \mDefine \left( \phi_X,\phi_Y,\psi \right)$ consists of \emph{encoders} $\phi_X$ and $\phi_Y$ which are mappings from the set of length-$n$ source vectors to the messages sets $\mset{M}_X$ and $\mset{M}_Y$:
\begin{subequations}
\begin{align}
&\phi_X: \sX^n \mapsto \mset{M}_X
\\
&\phi_Y: \sY^n \mapsto \mset{M}_Y.
\end{align}
\end{subequations}
and a \emph{\decider{}}, which is a mapping from the set of possible message pairs to one of the hypotheses:
\begin{align}
\psi: \mset{M}_X \times \mset{M}_Y \mapsto \BF.
\end{align}

\end{definition}



\begin{definition}
For a given scheme $\scheme$, denote the decision given the pair $(\XX,\YY)$ by
\begin{align}
\hat{H} \mDefine \psi\left(\phi_X(\XX),\phi_Y(\YY)\right).
\end{align}
The \emph{decision error probabilities} of $\scheme$ are given by
\begin{subequations}
\begin{align}
&\epsilon_i \mDefine \pr{ \hat{H} \neq H \mid H=i}, \mspcd i=0,1.
\end{align}
\end{subequations}
\end{definition}


\begin{definition}
For any $E_0>0$ and $E_1>0$, the exponent pair $(E_0,E_1)$ is said to be \emph{achievable} at rates $(R_X,R_Y)$ if there exists a sequence of schemes
\begin{align}
\scheme^{(n)} \mDefine \left( \phi_X^{(n)},\phi_Y^{(n)},\psi^{(n)} \right), \mspcd n=1,2,\ldots
\end{align}
with corresponding sequences of message sets $\mset{M}_X^{(n)}$ and $\mset{M}_Y^{(n)}$ 
and error probabilities $\epsilon_i^{(n)}$, $i \in \{0,1\}$, 
such that\footnote{All logarithms are taken to the base 2, and all rates are in units of bits per sample.}
\begin{subequations}
\begin{align}
&\limsup_{n \mGoesTo \infty} \frac{1}{n} \log \card{\mset{M}_X^{(n)}} \leq R_X
\\
&\limsup_{n \mGoesTo \infty} \frac{1}{n} \log \card{\mset{M}_Y^{(n)}} \leq R_Y,
\end{align}
and
\begin{align}
&\liminf_{n \mGoesTo \infty} - \frac{1}{n} \log \epsilon_i^{(n)} \geq E_i, \mspcd i=0,1.
\end{align}
\end{subequations}
The achievable exponent region $\DhtAchv{}(R_X,R_Y)$ is the closure of the set of all achievable exponent pairs.\footnote{For simplicity of the notation we omit here and in subsequent definitions the explicit dependence on the distributions $(\Po{},\Pl{})$.}
\end{definition}

The case where only one of the error probabilities decays exponentially is of special interest; we call the resulting quantity the \emph{Stein exponent} after Stein's Lemma (see, e.g.,~\cite[Chapter 12]{CoverBook}). When $\epsilon_1^{(n)}$ is exponential, the Stein exponent is  defined as:
\begin{align}
\steinE_1(R_X,R_Y) &\mDefine \sup_{E_0 > 0} \left\{E_1: \exists (E_0,E_1) \in \DhtAchv{}(R_X,R_Y)\right\}.
\end{align}
$\steinE_0(R_X,R_Y)$ is defined similarly.

We will concentrate on this work on two special cases of rate constraints, where for simplicity we can make the notation more concise.
\begin{enumerate}
\item One-sided constraint where $R_Y=\infty$. We shall denote the achievable region and Stein exponents as $\DhtAchvX{}(R_X)$, $\steinEoX{}(R_X)$ and $\steinElX{}(R_X)$. 
\item Symmetric constraint where $R_X=R_Y=R$. We shall denote the achievable region and Stein exponents as $\DhtAchvSym{}(R)$, $\steinEoSym{}(R)$ and $\steinElSym{}(R)$. 
\end{enumerate}
Note that for any $R$ we have that $\DhtAchvSym{}(R) \subseteq \DhtAchvX{}(R)$.

Whenever considering a specific source distribution, we will take $(X,Y)$ to be a DSBS. Recalling \eqref{eq:hypothses}, that means that $X$ and $Y$ are binary symmetric, and the ``noise'' $Z \mDefine Y \ominus X$ satisfies:
\begin{align}
\pr {Z=1 | H=i} = p_i, \mspcd i = 0,1
\end{align}
for some parameters $0 \leq p_0 \leq p_1 \leq 1/2$ (note that there is loss of generality in assuming that both probabilities are on the same side of $1/2$).

\subsection{Further Notations}


The following notations of probability distribution functions are demonstrated for random variables $X, Y$ and $Z$ over alphabets $\sX, \sY$ and $\sZ$, respectively. The probability distribution function of a random variable $X$ is denoted by $P_X$, and the conditional probability distribution function of a random variable $Y$ given a random variable $X$ is denoted by $P_{Y|X}$.
A composition $P_X$ and $P_{Y|X}$ is denoted by $P_X P_{Y|X}$, leading to the following joint probability distribution function $P_{X,Y}$ of $X$ and $Y$:
\begin{align}
\pa{P_X P_{Y|X}}(x,y) \mDefine P_X(x) P_{Y|X}(y|x),
\end{align}
for any pair $x \in \sX$ and $y \in \sY$.


The Shannon entropy of a random variable $X$ is denoted by $\H{P_X}$, and the \KL{} divergence of a pair of probability distribution functions $(P,Q)$ is denoted by $\D{P}{Q}$. The mutual information of a pair of random variables $(X,Y)$ is denoted by $\ckI{P_X}{P_{Y|X}}$.
The similar conditional functionals of the entropy, divergence and mutual information are defined by an expectation over the a-priori distribution: the conditional entropy of a random variable $X$ given a random variable $Z$ is denoted by
\begin{align}
\Hcond{P_{X|Z}}{P_Z} \mDefine \sum_{x\in\sX} P_X(x) \sum_{y\in\sY} P_{Y|X}(y|x) \log \frac{1}{P_{Y|X}(y|x)}.
\end{align}
The divergence of a pair of conditional probability distribution functions $P_{X|Z}$ and $P_{Y|Z}$ is denoted by $$\Dcond{P_{X|Z}}{P_{Y|Z}}{P_Z}.$$
The conditional mutual information of a pair of random variables $(X,Y)$ given a random variable $Z$ is denoted by $$\ckIcond{P_{X|Z}}{P_{Y|X,Z}}{P_Z},$$ and notice that it is equal to
$$\Hcond{P_{X|Z}}{P_Z} - \Hcond{P_{X|Y,Z}}{P_Z P_{X|Z}}.$$
If there is a Markov chain $Z \markov X \markov Y$, then we can omit the $Z$ from $P_{Y|X,Z}$ and the expression becomes $$\ckIcond{P_{X|Z}}{P_{Y|X}}{P_Z}.$$

Since we concentrate on a binary case, we need the following. 
Denote the \emph{binary divergence} of a pair $(p,q)$, where $p,q \in (0,1)$, by
\begin{align}
\dB{p}{q} &\mDefine p \log \frac{p}{q} + (1-p) \log \frac{1-p}{1-q},
\end{align}
which is the \KL{} divergence of the pair of probability distributions $((p,1-p),(q,1-q))$.
Denote the \emph{binary entropy} of $p \in (0,1)$ by
\begin{align}
\hB{p} &\mDefine p \log \frac{1}{p} + (1-p) \log \frac{1}{1-p},
\end{align}
which is the entropy function of the probability distribution $(p,1-p)$.
%
%
Denote the \GV{} relative distance of a code of rate $R$, $\deltaGV: [0,1] \mapsto [0,1/2]$ by
\begin{align}
\deltaGV{}(R) \mDefine \hBinv{1-R}.
\end{align}

The operator $\oplus$ denotes addition over the binary field. The operator $\ominus$ is equivalent to the $\oplus$ operator over the binary field, but nevertheless, we keep them for the sake of consistency.

%
%
The \emph{Hamming weight} of a vector $\uu=(u_1,\ldots,u_n) \in \BF^n$ is denoted by
\begin{align}
\wH{\uu} = \sum_{k=1}^n \indicator{u_i=1},
\end{align}
where $\indicator{\cdot}$ denotes the indicator function, and the sum is over the reals.
The \emph{normalized Hamming weight} of this vector is denoted by
\begin{align}
\nwH{\uu} = \frac{1}{n} \wH{\uu}.
\end{align}
%
%
Denote the $n$ dimensional Hamming ball with center $\cc$ and normalized radius $r \in [0,1]$ by
\begin{align}
\Ball{n}{\cc}{r} \mDefine \left\{ \xx \in \BF^n \mid \nwH{\xx \ominus \cc} \leq r \right\},
\end{align}
%
%
The \emph{binary convolution} of $p,q \in [0,1]$ is defined by
\begin{align}
p*q \mDefine (1-p)q + p(1-q). 
\end{align}


\begin{definition}[Bernoulli Noise]
\label{def:Bernoulli_noise}
A Bernoulli random variable $Z$ with $\pr{Z=1}=p$ is denoted by $Z\Ber{p}$.
An $n$ dimensional random vector $\ZZ$ with i.i.d. entries $Z_i \Ber{p}$ for $i=1,\ldots,n$ is called a \emph{Bernoulli noise}, and denoted by
\begin{align}
\ZZ \BerV{n}{p} 
\end{align}
\end{definition}

\begin{definition}[Fixed-Type Noise]
\label{def:fixed_type_noise}
Denote the set of vectors with \emph{type} $a\in[0,1]$ by
\begin{subequations}
\begin{align}
\sT_n(a) \mDefine \{\xx\in\BF^n : \nwH{\xx} = a\}.
\end{align}
A noise
\begin{align}
\mvec{N} \Uniform{ \Type{n}{a} } 
\end{align}
\end{subequations}
is called an $n$-dimensional \emph{fixed-type noise} of type $a \in [0,1]$.
\end{definition}


For any two sequences, $\{a_n\}_{n=1}^\infty$ and $\{b_n\}_{n=1}^\infty$, we write $a_n \doteq b_n$ if $\lim_{n \mGoesTo \infty} n^{-1}  \log (a_n/b_n) = 0$.
We write $a_n \dotleq b_n$ if $\lim_{n \mGoesTo \infty} n^{-1} \log (a_n/b_n) \leq 0$.

For any two sequences of random vectors $\XX_n, \YY_n \in \sX^n$ ($n=1,2,\ldots$), we write
\begin{align}
\XX_n \doteqD \YY_n
\end{align}
if
\begin{align}
\pr{\XX_n = \xx_n} \doteq \pr{\YY = \xx_n}
\end{align}
uniformly over $\xx_n \in \sX^n$, that is,
\begin{align}
\lim_{n \mGoesTo \infty} \frac{1}{n} \log \frac{P_{\XX_n}(\xx_n)}{P_{\YY_n}(\xx_n)} = 0
\end{align}
uniformly over $\xx_n \in \sX^n$.
We write $\XX_n \dotleqD \YY_n$ if 
\begin{align}
\lim_{n \mGoesTo \infty} \frac{1}{n} \log \frac{P_{\XX_n}(\xx_n)}{P_{\YY_n}(\xx_n)} \leq 0
\end{align}
uniformly over $\xx_n \in \sX^n$.

The set of non-negative integers are denoted by $\mZZ_+$, and the set of natural numbers, i.e., $1,2,\ldots$, by $\mNN$.

\subsection{Some Basic Results}
\label{sec:basic}
When the rate is not constrained, the decision function has access to the full source sequences. The optimal tradeoff of the two types of errors is given by the following decision function, depending on the parameter $T \geq 0$ (Neyman-Pearson~\cite{NeymanPearson1933}),\footnote{In order to achieve the full Neyman-Pearson tradeoff, special treatment of the case of equality is needed. As this issue has no effect on error exponents, we ignore it.}
\begin{align}
\varphi(\xx,\yy) = \left\{
\begin{array}{ll}
0, & P^{(0)}_{X,Y}(\xx,\yy) \geq T \cdot P^{(1)}_{X,Y}(\xx,\yy)
\\
1, & \text{otherwise}.
\end{array}
\right.
\end{align}

\begin{proposition}[Unconstrained Case]
Consider the hypothesis testing problem as defined in Section~\ref{sec:problem_statement}, where there is no rate constraint, i.e. $R_X=R_Y=\infty$, then $(E_0,E_1) \in \DhtAchvSym{}(\infty)$ if and only if there exists a distribution function $\PH{(*)}{X,Y}$ over the pair $(\mset{X},\mset{Y})$ such that \begin{align}E_i \leq \D{\PH{(*)}{X,Y}}{\PH{(i)}{X,Y}}, \text{ for } i=0,1.\end{align}
\end{proposition}

For proof, see e.g. \cite{CoverBook}. Note that in fact rates equal to the logarithms of the alphabet sizes suffice. 

For the DSBS, the Neyman-Pearson \decider{} is a threshold on the weight of the noise sequence. We denote it (with some abuse of notations)  by $$\varphi_\thd(\xx,\yy) \mDefine \varphi_\thd\left( \nwH{\xx \oplus \yy} \right),$$ where $\varphi_\thd:\mRR \mapsto \BF$ is a threshold test,
\begin{align}
\label{eq:optimal_decision}
\varphi_\thd(w) = \left\{
\begin{array}{ll}
0, & w \leq \thd
\\
1, & w > \thd.
\end{array}
\right.
\end{align}
It leads to the following performance:
\begin{corollary}[Unconstrained Case, DSBS]
\label{cor:unconstrained:DSBS}
For the DSBS, $\DhtAchvX{}(1) = \DhtAchv{}(1)$, and they consist of all pairs $(E_0,E_1)$ satisfying that for some
$\thd \in (p_0,p_1)$,
\begin{align} E_i \leq \dB{\thd}{p_i}, \text{ for } i=0,1. \end{align}
\end{corollary}

We now note a time-sharing result, which is general to any given achievable set.
\begin{proposition}[time-sharing] \label{prop:time-sharing:general}
Suppose that $(E_0,E_1) \in \DhtAchv{}(R_X,R_Y)$. Then $\forall \alpha\in[0,1]:$
\begin{align}(\alpha E_0,\alpha E_1) \in \DhtAchv{}(\alpha R_X, \alpha R_Y).\end{align}
\end{proposition}

The proof is standard, by noting that any scheme may be applied to an $\alpha$-portion of the source blocks, ignoring the additional samples. Applying this technique to Corollary~\ref{cor:unconstrained:DSBS}, we have a simple scheme where each encoder sends only a fraction of its observed vector.

\begin{corollary} \label{cor:straightforward_approach:DSBS}
Consider the DSBS hypothesis testing problem as defined in Section~\ref{sec:problem_statement}.
For any rate constraint $R \in [0,1]$, for any $\thd \in (p_0,p_1)$
\begin{align} \label{eq:baseline_exponents:DSBS} \pa{R \cdot \dB{\thd}{p_0},R\cdot \dB{\thd}{p_1}} \in \DhtAchv{}(R) \end{align}
\end{corollary}

Specializing to Stein's exponents, we have: 
\begin{subequations}
\begin{align}
&\steinE_0(R) \geq R \cdot \dB{p_1}{p_0}
\\
&\steinE_1(R) \geq R \cdot \dB{p_0}{p_1},
\end{align}
\label{cor:straightforward_approach:DSBS:Stein}
\end{subequations}
Of course we may apply the same result to the one-sided constrained case, i.e., $\DhtAchvX{}$ and the corresponding Stein exponents.

%
\section{One-Sided Constraint: Previous Results}
%
\label{sec:related_results}

In this section we review previous results for the one-sided constraint case $R_Y=\infty$. We first present them for general distributions $\Po{X,Y},\Pl{X,Y}$ and then specialize to the DSBS.

\subsection{General Sources}
\label{sec:one_sided:general_sources}

Ahlswede and \Csiszar{} have established the following achievable Stein's exponent.

\newcommand{\ACauxY}{\ensuremath{ \HypVar{Y}{*} }} 
\begin{proposition}[{\cite[Theorem 5]{AhlswedeCsiszar1986}}]
\label{prop:AC86}
For any $R_X>0$,
\begin{align} \label{eq:AC86}
\steinElX\pa{R} &\geq \D{\Po{X}}{\Pl{X}} \nonumber\\&
+ 
\max_{
\begin{subarray}{c}
\\
P_{V|X}:\\
\ckI{\Po{X}}{P_{V|X}} \leq R_X
\end{subarray}
}
\D{\Po{V,Y}}{\PH{(*)}{V,Y}},
\end{align}
where $\Po{V,Y}$ and $\PH{(*)}{V,Y}$ are the marginals of $$\Po{V,X,Y} \mDefine P_{V|X} \Po{X} \Po{Y|X}$$ and $$\PH{(*)}{V,X,Y} \mDefine P_{V|X} \Po{X} \Pl{Y|X},$$ respectively.
\end{proposition}
The first term of \eqref{eq:AC86} reflects the contribution of the type of $\XX$ (which can be conveyed with zero rate), while the second reflects the contribution of the lossy version of $\XX$ sent with rate $R_X$. Interestingly, this exponent is optimal for case $\Pl{Y|X} = \Pl{Y}$, known as test against independence.

Han has improved upon this exponent by conveying the joint type of the source sequence $X$ and its quantized version (represented by $V$) to the decision function.\footnote{Han's result also extends to any rate pair $(R_X,R_Y)$; however, we only state it for the single-sided constraint.}


\begin{proposition}[{\cite[Theorems 2,3]{Han1987}}]
\label{prop:Han87}
For any $R_X \geq 0$,
\begin{subequations}
\begin{align}
&\steinElX\pa{R_X} \geq
\D{\Po{X}}{\Pl{X}} + \max_{
\begin{subarray}{c}
P_{V|X}:\\
\ckI{\Po{X}}{P_{V|X}} \leq R_X,\\ \card{\mset{V}} \leq \card{\mset{X}}+1
\end{subarray}}
\sigma_\text{HAN}(V),
\end{align}
where
\begin{align}
\sigma_\text{HAN}(V) &\mDefine \min_{
\begin{subarray}{c}
\Pt{Y|V,X}:\\
\Pt{V,Y}=\Po{V,Y}
\end{subarray}}
\Dcond{\Pt{Y|X,V}}{\Pl{Y|X}}{\Po{X} P_{V|X}}
\end{align}
\end{subequations}
and where $\Po{V,Y}$ and $\Pt{V,Y}$ are the marginals of
\begin{subequations}
\begin{align}
\label{eq:han_composed_distributions_o}
\Po{V,X,Y} \mDefine P_{V|X} \Po{X} \Po{Y|X}
\end{align}
and
\begin{align}
\label{eq:han_composed_distributions_t}
\Pt{V,X,Y} \mDefine P_{V|X} \Po{X} \Pt{Y|V,X},
\end{align}
\end{subequations}
respectively.
\end{proposition}


The following result by Shimokawa et al., gives a tighter achievable bound by using the side information $\YY$ when encoding $\XX$.
\begin{proposition}[{\cite[Corollary III.2]{ShimokawaHanAmari1994ISIT},\cite[Theorem 4.3]{HanAmari1998}}]
\label{prop:SHA94}
\begin{subequations}
Define
\begin{align}
\label{eq:SHA_rho}
&\sigma_\text{SHA}(V) \mDefine  -\ckIcond{\Po{X|Y}}{P_{V|X}}{\Po{Y}}
\nonumber\\&+ \min_{
\begin{subarray}{c}
\Pt{Y|V,X}:\\
\Pt{Y} = \Po{Y},\\
\Hcond{\Pt{V|Y}}{P_\tY} \geq \Hcond{\Po{V|Y}}{\Po{Y}}
\end{subarray}
} \Dcond{\Pt{Y|X,V}}{\Pl{Y|X}}{\Po{X} P_{V|X}},
\end{align}
where $\Po{V,Y}$ and $\Pt{V,Y}$ are the marginals of the distributions defined in~\eqref{eq:han_composed_distributions_o} and~\eqref{eq:han_composed_distributions_t},
respectively.
Then, for any $R_X>0$,
\begin{align}
&\steinElX\pa{R_X} \geq \D{\Po{X}}{\Pl{X}}
\nonumber\\&
+ \max_{
\begin{subarray}{c}
P_{V|X}:\\
\ckIcond{\Po{X|Y}}{P_{V|X}}{\Po{Y}} \leq R_X,\\
\scriptstyle
\card{\mset{V}} \leq \card{\mset{X}}+1
\end{subarray}
}
\min \left\{ \sigma_\text{HAN}(V),R_X+\sigma_\text{SHA}(V) \right\}.
\end{align}
\end{subequations}
\end{proposition}
Notice that for $P_{V|X}$ such that $\ckI{\Po{X}}{P_{V|X}} \leq R_X$, the bound of the last proposition will be not greater than the bound of Proposition~\ref{prop:Han87}. Therefore the overall bound yields by taking the maximal one.

It is worth pointing out that for distributed rate-distortion problem, the bound in Proposition~\ref{prop:SHA94} is in general suboptimal~\cite{WagnerKellyAltig2011}.


A non-trivial outer bound derived by Rahman and Wagner~\cite{RahmanWagner2012} using an additional information at the decoder, which does not exist in the original problem.
\begin{proposition}[{\cite[Corollary 5]{RahmanWagner2012}}]
\label{prop:RW}
Suppose that
\begin{subequations}
\begin{align}
\Po{X} = \Pl{X}.
\end{align}
Consider a pair of conditional distributions $\Po{Z|X,Y}$ and $\Pl{Z|X,Y}$ such that
\begin{align}
\Po{Z|X} = \Pl{Z|X}
\end{align}
and such that $X \markov Z \markov Y$ under the distribution
\begin{align}
\Pl{X,Y,Z} \mDefine \Pl{X,Z} \Pl{Z|X,Y}.
\end{align}
Then, for any $R_X>0$,
\begin{align}
&\steinElX(R_X) \leq 
D\pa{\Po{Y|Z} \| \Pl{Y|Z} \mid P_Z}
\nonumber\\&
+ \max_{
\begin{subarray}{c}
P_{V|X}:\\
\ckIcond{\Po{X|Y}}{P_{V|X}}{\Po{Y}} \leq R_X,\\
\scriptstyle
\card{\mset{V}} \leq \card{\mset{X}}+1
\end{subarray}
}
\ckIcond{\Po{Y|Z}}{\Po{X|Y,Z}P_{V|X}}{P_Z}.
\end{align}
\end{subequations}
\end{proposition}

\subsection{Specializing to the DSBS}
\label{sec:related_results:specializing_to_DBSS}

We now specialize the results of Section~\ref{sec:one_sided:general_sources} to the DSBS. Throughout, we choose the auxiliary variable $V$ to be connected to $X$ by a binary symmetric channel with crossover probability $a$; with some abuse of notation, we write e.g. $\sigma(a)$ for the specialization of $\sigma(V)$. Due to symmetry, we conjecture that this choice of $V$ is optimal, up to time sharing that can be applied according to Proposition~\ref{prop:time-sharing:general}; we do not explicitly write the time-sharing expressions.

The connection between the general and DSBS-specific results can be shown;.However, we follow a direction that is more relevant to this work, providing for each result a direct interpretation, explaining how it can be obtained for the DSBS; in doing that, we follow the interpretations of Rahman and Wagner~\cite{RahmanWagner2012}.


The Ahlswede-\Csiszar{} scheme of Proposition~\ref{prop:AC86} amounts to quantization of the source $\XX$, without using $\YY$ as side information.   
\begin{corollary}[Proposition~\ref{prop:AC86}, DSBS with symmetric auxiliary]
\label{cor:AC86:binary_source}
For any $R_X>0$,
\begin{subequations}
\begin{align}
\steinElX(R_X) \geq \sigma_\text{AC}(\deltaGV{}(R_X)),
\end{align}
where
\begin{align}
\sigma_\text{AC}(a) \mDefine \dB{a*p_0}{a*p_1}.
\end{align}
\end{subequations}
\end{corollary}
This performance can be obtained as follows. The encoder quantizes $\XX$ using a code that is rate-distortion optimal under the Hamming distortion measure; specifically, averaging over the random quantizer, the source and reconstruction are jointly distributed according to the RDF-achieving test channel, that is, the reconstruction $\hat\XX$ is obtained from the source $\XX$ by a BSC with crossover probability $a$ that satisfies the RDF, namely $a=\deltaGV{}(R_X)$. The decision function is $\phi_\thd(\hat\XX,\YY)$ which can be seen as two-stage: first the source difference sequence is estimated as $\hat\ZZ = \YY \ominus \hat\XX$, and then a threshold is applied to the weight of that sequence, as if it were the true noise. Notice that given $H=i$, $\hat\ZZ \BerV{n}{a*p_i}$; the exponents are thus the probabilities of such a vector to fall inside or outside a Hamming sphere of radius $n\thd$ around the origin. As Proposition~\ref{prop:AC86} relates to a Stein exponent, the threshold $\thd$ is set arbitrarily close to $a*p_0$, resulting in the following; one can easily generalize to an achievable exponent region.

The Han scheme of Proposition~\ref{prop:Han87} amounts (for the DSBS) to a similar approach, using a more favorable quantization scheme. In order to express its performance, we use the following exponent, which is explicitly evaluated in Appendix~\ref{sec:exp_hamming_ball}. While it is a bit more general than what we need at this point, this definition will allow us to present later results in a unified manner.

\begin{definition}
\label{def:EBTinball}
Fix some parameters $p,a,\thd,w \in [0,1]$. Let $\cc_n \in \BF^n, n=1,2,\ldots$ be a sequence of vectors such that $\lim_{n \mGoesTo \infty} \nwH{\cc_n}=w$.
Let $\ZZ \BerV{n}{p}$ and let $\UU \Uniform{\Type{n}{a}}$. 
Then:
\begin{align}
\label{eq:EBYinball}
\EBTinBall{p}{a}{w}{\thd} \triangleq
\lim_{n\rightarrow\infty} - \frac{1}{n} \log \pr{\ZZ \oplus \UU \in \Ball{n}{\cc_n}{\thd}}. 
\end{align}
\end{definition}

\begin{corollary}[Proposition~\ref{prop:Han87}, DSBS with symmetric auxiliary]
\label{cor:Han87:binary_source}
For any $R_X>0$,
\begin{subequations}
\begin{align}
\steinElX(R_X) \geq \sigma_\text{HAN}(\deltaGV{}(R_X)),
\end{align}
where
\begin{align}
\sigma_\text{HAN}(a) \mDefine \EBTinBall{p_1}{a}{0}{a*p_0}.
\end{align}
\end{subequations}
\end{corollary}

One can show that $\sigma_\text{HAN}(a) \geq \sigma_\text{AC}(a)$, where the inequality is strict for all $p_1<1/2$ (recall that for $p_1=1/2$, ``testing against independence'', the Alswhede-\Csiszar{} scheme is already optimal). The improvement comes from having quantization error that is  fixed-type $a$ (recall Definition~\ref{def:fixed_type_noise}) rather than Bernoulli. Thus, $\hat\ZZ$ is ``mixed'' uniform-Bernoulli; the probability of that noise to enter a ball around the origin is reduced with respect to that of the Bernoulli $\hat\ZZ$ of Corollary~\ref{cor:AC86:binary_source}.


The Shimokawa et al. scheme of Proposition~\ref{prop:SHA94} is similar in the DSBS case, except that the compression of $\XX$ now uses side-information. Namely, Wyner-Ziv style binning is used. When the bin is not correctly decoded, a decision error may occur. The resulting performance is given in the following.
\begin{corollary}[Proposition~\ref{prop:SHA94}, DSBS with symmetric auxiliary]
\label{cor:SHA94:DSBS}
For any $R_X>0$,
\begin{subequations}
\begin{align}
\label{eq:SHAs_expression}
\steinElX(R_X) \geq 
\max_{0 \leq a \leq \deltaGV{}(R_X) } \min \pc{\sigma_\text{HAN}(a), \sigma_\text{SHA}(R_X,a)},
\end{align}
where
\begin{align}
\sigma_\text{SHA}(R,a) \mDefine R - \hB{a*p_0} + \hB{a} \label{eq:SHA_DSBS_rho}
\end{align}
\end{subequations}
\end{corollary}

This exponent can be thought of as follows. The encoder performs fixed-type quantization as in Han's scheme, except that the quantization type $a$ is now smaller than $\deltaGV{}(R_X)$. The indices thus have rate $1-H_b(a)$. Now these indices are distributed to bins; as the rate of the bin index is $R_X$, each bin is of rate $1-H_b(a)-R_X$. The decision function decodes the bin index using the side information $\YY$, and then proceeds as in Han's scheme. 

The two terms in the minimization~\eqref{eq:SHAs_expression} represent the sum of the events of decision error combined with bin-decoding success and error, respectively. The first is as before, hence the use of $\sigma_\text{SHA}$. For the second, it can be shown that as a worst-case assumption, $\hat\XX$ resulting from a decoding error is uniformly distributed over all binary sequences. By considering volumes, the exponent of the probability of the reconstruction to fall inside an $n\thd$-sphere is thus at most $1-H_b(\thd)$; a union bound over the bin gives  $\sigma_\text{SHA}$. 

\begin{remark} It may be better not to use binning altogether (thus avoiding binning errors), i.e., the exponent of Corollary~\ref{cor:SHA94:DSBS} is not always higher than that of  Corollary~\ref{cor:Han87:binary_source}. \end{remark}

\begin{remark} An important special case of this scheme is when lossless compression is used, and Wyner-Ziv coding reduces to a side-information case of Slepian-Wolf coding. This amounts to forcing $a=0$. If no binning error occurred, we are in the same situation as in the unconstrained case. Thus,  we have the exponent:
\begin{subequations}
\begin{align}   \min  \left(  D_b(p_0\|p_1),  \sigma_{\text {SHA}}(R_X) \right) , \end{align}
where
\begin{align} 
\sigma_{\text {SHA}}(R) \mDefine \sigma_{\text {SHA}}(R,0) =  R - H_b(p_0). \end{align}
\end{subequations}
\end{remark}
 
We have seen thus that various combinations of quantization and binning; Table~\ref{tbl:possible_schemes} summarizes the different possible schemes. 

\begin{table}
\label{tbl:possible_schemes}
\centering
\begin{tabular}{l|c|c}
\hspace{0.5cm} \bf Coding component & \bf Lossless & \bf Lossy\\
\hline
\bf Oblivious to $Y$ & TS & Q + TS~\cite{AhlswedeCsiszar1986,Han1987}\\
\hline
\bf Using side-information $Y$ & Bin + TS & Q + Bin + TS~\cite{ShimokawaHanAmari1994ISIT}
\end{tabular}
\caption{Summary of possible schemes. TS stands for time-sharing, Q stands for quantization, Bin stands for binning.}
\end{table}


An upper bound is obtained  by specializing the Rahman-Wagner outer bound of Proposition~\ref{prop:RW} to the DSBS.

\begin{corollary}[Proposition~\ref{prop:RW}, DSBS with symmetric additional information]
\label{cor:RW_Binary}
\begin{align}
&\sigma_\text{RW}(R,\zeta,a) \mDefine \min_{0 \leq \zeta \leq p_0} \min_{0 \leq b_1 \leq 1}\max_{0 \leq a \leq 1/2}
\nonumber\\&\hB{\gamma}
- \zeta*a \hB{\frac{b_1 \cdot \zeta \cdot (1-a) + b_0 \cdot (1-\zeta) \cdot a}{\zeta*a}}
\nonumber\\&
+ (1-\zeta*a) \hB{\frac{b_1 \cdot \zeta \cdot a + b_0 \cdot (1-\zeta) \cdot (1-a)}{1-\zeta*a}}
\nonumber\\&
\dB{b_0 \cdot (1-a) + b_1 \cdot a}{\frac{p_0-a}{1-2a}},
\end{align}
where $b_0 \mDefine \frac{p_1 - a \cdot (1-b_1)}{1-a}$.
\end{corollary}
We note that it seems plausible that the exponent for $p_1=1/2$, given by   Corollary~\ref{cor:Han87:binary_source}, is an upper to general $p_1$, i.e., 
\begin{align}
\steinElX\pa{R_X} \leq 1 - \hB{p_0*\deltaGV{}(R)}.
\end{align}

Next we compare the performance of these coding schemes in order to understand the effect of each of the different components of the coding schemes on the performance.

\section{Background: Linear Codes and Error Exponents}
\label{sec:linear_codes_and_EE}

In this section we define code ensembles that will be used in the sequel, and present their properties. Although the specific properties of linear codes are not required until Section~\ref{sec:symmetric_rate_constraint}, we put an emphasis on such codes already; this simplifies the proofs of some properties we need to show, and also helps to present the different results in a more unified manner.

\subsection{Linear Codes}


\begin{definition}[Linear Code]
We define a \emph{linear code} via a $k \times n$ generating matrix $\mat{G}$ over the binary field.
This induces the  linear codebook:
\begin{align}
\label{eq:linear_code}
\sC = \{ \cc: \cc = \uu \mat{G}, \mspcc \uu\in\BF^k \},
\end{align}
where $\mvec{u}\in\BF^k$ is a row vector.
\end{definition}
Assuming that all rows of $\mat{G}$ are linearly independent, there are $2^k$ codewords in $\sC$, so the code rate is
\begin{align}
\label{eq:rate:linear_code}
R = \frac{k}{n}.
\end{align}
Clearly, for any rate (up to $1$), there exists a linear code of this rate asymptotically as $n \mGoesTo \infty$.

A linear code is also called a \emph{parity-check code}, and may be specified by a $(n-k) \times n$ (binary) parity-check matrix $\mat{H}$. The code $\sC$ contains all the $n$-length binary row vectors $\cc$ whose \emph{syndrome} 
\begin{align} \label{eq:syndrome} \ss \mDefine \cc \mat{H}^T \end{align} is equal to the $n-k$ all zero row vector, i.e.,
\begin{align}
\sC \mDefine \left\{ \cc \in \BF^n : \cc \mat{H}^T= \mvec{0} \right\}.
\end{align}


Given some general syndrome $\ss\in\BF^{n-k}$, denote the \emph{coset} of $\ss$ by
\begin{align}
\sC_\ss \mDefine \{\xx\in\BF^n : \xx \mat{H}^T = \ss\}.
\end{align}

The minimum Hamming distance \emph{quantizer} of a vector $\xx\in\BF^n$ with respect to a code $\sC\subseteq\BF^n$ is given by 
\begin{align}
Q_{\sC}(\xx) \mDefine \arg \min_{\cc\in\sC} \nwH{\xx \ominus \cc}.
\label{eq:quant_def}
\end{align}

For any syndrome $\ss$ with respect to the code $\sC$, the decoding function $f_\sC(\ss):\BF^{n-k}\mapsto\BF^n$ gives the \emph{coset leader}, the   minimum Hamming weight vector within the coset of $\ss$:
\begin{align}
f_\sC(\ss)  &\mDefine \arg \min_{\zz\in\sC_\ss} \nwH{\zz}. 
\end{align}


Maximum-likelihood decoding of a parity-check code, over a BSC $Y = X \oplus Z$, amounts to syndrome decoding $\hat{\xx} = \yy \ominus f_\sC(\yy)$~\cite[Theorem 6.1.1]{GallagerBook1968}. 
The basic ``Voronoi'' set is given by
\begin{align}
\Omega_\mvec{0} \mDefine \left\{ \zz : \zz \ominus f_\sC(\zz \mat{H}^T) = \mvec{0} \right\}.
\end{align}
The ML decision region of any codeword $\cc\in\sC$ is equal to a translate of $\Omega_0$, i.e.,
\begin{align}
\Omega_\cc &\mDefine \left\{ \yy : \yy \ominus f_\sC(\yy \mat{H}^T) = \cc \right\}
\\&= \Omega_\mvec{0} + \cc.
\end{align}

\subsection{Properties of Linear Codes}


\begin{definition}[Normalized Distance Distribution]
\label{def:distance_spectrum}
The \emph{normalized distance (or weight) distribution} of a linear code $\sC$ for a parameter $0 \leq w \leq 1$ is defined to be the fraction of codewords $\cc\neq\mvec{0}$, with normalized weight at most $w$, i.e.,
\begin{align}
\dspectrum{\sC}{w} \mDefine \frac{1}{ \card{\sC}} \sum_{\cc \in \sC\setminus\{\mvec{0}\}} \indicator{\nwH{\cc} \leq w},
\end{align}
where $\indicator{\cdot}$ is the indicator function.
\end{definition}

\begin{definition}[Normalized Minimum Distance]
\label{def:deltamin}
The \emph{normalized minimum distance} of a linear code $\sC$ is defined as
\begin{align}
\deltamin(\sC) \mDefine \min_{\cc\in\sC\setminus\{\mvec{0}\}} \nwH{\cc}
\end{align}
\end{definition}

\begin{definition}[Normalized Covering Radius]
\label{def:covering_radius}
The \emph{normalized covering radius} of a code $\sC \in \BF^n$ is the smallest integer such that every vector $\xx \in \BF^n$ is covered by a Hamming ball with radius $r$ and center at some $\cc \in \sC$, normalized by the blocklength, i.e.:
\begin{align}
\NrCoveringRadius(\sC) \mDefine \max_{\xx \in \BF^n} \min_{\cc \in \sC} \nwH{\xx\ominus\cc}.
\end{align}
\end{definition}

\begin{definition}[Normalized Packing Radius]
\label{def:packing_radius}
The \emph{normalized packing radius} of a linear code $\sC$ is defined to be half the normalized minimal distance of its codewords, i.e.,
\begin{align}
\NrPackingRadius(\sC) \mDefine \frac{1}{2} \deltamin(\sC).
\end{align}
\end{definition}

\subsection{Good Linear Codes}

We need two notions of goodness of codes, as follows.

\begin{definition}[Spectrum-Good Codes]
A sequence of codes \SequenceOfCodes{} with rate $R$ is said to be \emph{spectrum-good} if for any $w \geq 0$,
\[ \dspectrum{\Cn}{w} \doteq \dspectrumgood{n}{R}{w} \]
where
\begin{align} \label{eq:spectrum-good}
\dspectrumgood{n}{R}{w} = \left\{
\begin{array}{ll}
2^{-n \dB{w}{1/2}}, & w > \deltaGV{}(R)
\\
0, & \text{otherwise}
\end{array}
\right.
.
\end{align}
\end{definition}

\begin{definition}[Covering-Good]
A sequence of codes \SequenceOfCodes{} with rate $R$ is said to be \emph{covering-good} if
\begin{align}
\label{eq:def:covering_good}
\NrCoveringRadius(\Cn) \mGoesToAs{n \mGoesTo \infty} \deltaGV(R)
\end{align}
\end{definition}

The existence of linear codes satisfying these properties is well known. Specifically, consider the ensemble of constructed by random generating matrices, where each entry of the matrix is drawn uniformly and statistically independent with all other entries, then almost all members have a spectrum close to~\eqref{eq:spectrum-good}, see e.g.~\cite[Chapter 5.6-5.7]{GallagerBook1968}; in addition, for almost all members, a process of appending rows to the generating matrix (with vanishing rate) results in a normalized covering radius close to $\deltaGV$~\cite[Theorem 12.3.5]{LitsynCoveringCodesBook1997}. These existence arguments are detailed in Appendix~\ref{app:good}.
We need the following properties of good codes.

Spectrum-good codes obtain the best known error exponent for the BSC. Namely, for a BSC with crossover probability $p$, they achieve $\BestKnownAchievableErrExp{p}{R}$, given by  
\begin{subequations}
\begin{align}
\label{eq:best_acheivable_errexp}
\BestKnownAchievableErrExp{p}{R} \mDefine \max \left\{ E_r(p,R), E_\text{ex}(p,R) \right\},
\end{align}
where
\begin{align}
E_r(p,R) &\mDefine \max_{\rho \in [0,1]} \rho - (1+\rho) \log \left( p^\frac{1}{1+\rho} + (1-p)^\frac{1}{1+\rho} \right) - \rho R
\end{align}
is the random-coding exponent, and
\begin{align}
E_\text{ex}(p,R) &\mDefine \max_{\rho \geq 1} -\rho \log\left( \frac{1}{2} + \frac{1}{2} \left[2 \sqrt{p(1-p)}\right]^\frac{1}{\rho} \right) - \rho R.
\end{align}
\end{subequations}
is the expurgated exponent. Notice that as the achievability depends only upon the spectrum, it is universal in $p$.

As for covering-good codes, we need the following result, shown that the quantization noise induced by covering-good codes is no worse than a noise that is uniform over an $n\deltaGV$-Hamming ball.
\begin{lemma}
\label{lem:quantization_uniform}
Consider a covering-good sequence of codes \SequenceOfCodes{} of rate $R$.
Then,
\begin{subequations}
\begin{align}
\XX \ominus Q_{\Cn}(\XX) \dotleqD \NN,
\end{align}
for
\begin{align}
\label{eq:quantization_good_X}
\XX &\Uniform{\BF^n}
\\
\label{eq:quantization_good_N}
\NN &\UCBall{n}{\NrCoveringRadius(\Cn)}.
\end{align}
\end{subequations}
Furthermore, the same holds when adding any random vector to both sides, i.e., for any $\ZZ$:
\begin{align}
\XX \ominus Q_{\Cn}(\XX) \oplus \ZZ \dotleqD \NN \oplus \ZZ ,
\end{align}
\end{lemma}
The proof appears in Appendix~\ref{app:covering_properties}.

\subsection{Nested Linear Codes} \label{subs:nested}
We briefly recall some basic definitions and properties related to nested linear codes. The reader is referred to \cite{ZamirShamaiErez02} for further details.


\begin{definition}[Nested Linear Code]
A \emph{nested linear code} with rate pair $(R_1,R_2)$ is a pair of linear codes $(\sC_1, \sC_2)$ with these rates, satisfying
\begin{align}
\label{eq:nested_def}
\sC_2 \subseteq \sC_1,
\end{align}
i.e., each codeword of $\sC_2$ is also a codeword of $\sC_1$ (see~\cite{ZamirShamaiErez02}). We call $\sC_1$ and $\sC_2$ \emph{fine code} and \emph{coarse code}, respectively.
\end{definition}

If a pair $\{(n,k_1),(n,k_2)\}$ of parity-check codes, $k_1 \geq k_2$, satisfies condition~\eqref{eq:nested_def}, then the corresponding parity-check matrices $\mat{H}_1$ and $\mat{H}_2$ are interrelated as
\begin{align}
\underbrace{\mat{H}_2^T}_{(n-k_2) \times n} =
[\underbrace{\mat{H}_1^T}_{(n-k_1) \times n}, \underbrace{\Delta \mat{H}^T_{\phantom{1}}}_{(k_1-k_2) \times n}],
\end{align}
where $\mat{H}_1$ is an $(n-k_1) \times n$ matrix, $\mat{H}_2$ is an $(n-k_2) \times n$ matrix, and $\Delta \mat{H}$ is a $(k_1-k_2) \times n$ matrix. This implies that the syndromes $\ss_1=\xx \mat{H}_1^T$ and $\ss_2=\xx \mat{H}_2^T$ associated with some $n$-vector $\xx$ are related as $\ss_2 = [\ss_1, \Delta \ss]$, where the length of $\Delta \ss$ is $k_1-k_2$ bits. In particular, if $\xx\in\sC_1$, then $\ss_2=[0,\ldots,0,\Delta\ss]$. We may, therefore, partition $\sC_1$ into $2^{k_1-k_2}$ cosets of $\sC_2$ by setting $\ss_1=\mvec{0}$, and varying $\Delta\ss$, i.e.,
\begin{align}
\sC_1 = \bigcup_{\Delta\ss\in\BF^{k_1-k_2}} \sC_{2,\ss_2}, \mspcd \text{where } \ss_2=[\mvec{0},\Delta\ss].
\end{align}
Finally, for a given pair of nested codes, the ``syndrome increment" $\Delta \ss$  is given by the function 
\begin{align}
\Delta \ss
&= \xx  \cdot \Delta \mat{H}^T.
\label{eq:syndrome_increment}
\end{align}
\begin{proposition}
\label{prop:syndrome}
Let the syndrome increment $\Delta \ss$ be computed for $\cc \in \sC_1$. Then, the coset leader corresponding to the syndrome of $\cc$ with respect to $\sC_2$ is given 
by 
\begin{align}
f_{\sC_2}(\cc \mat{H}_1^T)=f_{\sC_2}([\mvec{0},\Delta\ss]),
\end{align}
where $\mvec{0}$ is a zero row vector of length  $n-k_1$.
\end{proposition}
For a proof, see, e.g., \cite{ZamirShamaiErez02}.
\begin{definition}[Good Nested Linear Code] \label{def:good_nested}
A sequence of nested linear codes with rate pair $(R_1,R_2)$ is said to be good if the induced sequences of fine and coarse codes are covering-good and spectrum-good, respectively. 
\end{definition}

The existence of good nested linear codes follows naturally from the procedures used for constructing spectrum-good and covering-good codes; see Appendix~\ref{app:good} for a proof.

We need the following property of good nested codes. 
\begin{corollary}
\label{cor:nested_exponent}
Consider a sequence of good nested codes $(\Cn_1,\Cn_2)$, $n=1,2,\ldots$ with a rate pair $(R_1,R_2)$.
Let $\XX\Uniform{\BF^n}$ and $\ZZ\BerV{n}{p}$ be statistically independent. Denote $\UU \mDefine Q_{\Cn_1}(\XX)$, where quantization with respect to a code is
defined in \eqref{eq:quant_def}.
Then, 
\begin{align}
\pr{Q_{\Cn_{2,\SS}}(\XX \oplus \ZZ) \neq \UU} \dotleq
\pr{Q_{\Cn_{2,\SS}}\left(\UU \oplus \NN' \oplus \ZZ \right) \neq \UU}
\end{align}
where $\NN' \UCBall{n}{\deltaGV(R_1)}$ is statistically independent of $(\UU,\ZZ)$, where $\SS \mDefine \UU\mat{H}_2^T$, and where $\mat{H}_2$ is a parity check matrix of the coarse code $\Cn_2$.
\end{corollary}
The proof, which relies upon Lemma~\ref{lem:quantization_uniform}, is given in~Appendix~\ref{app:covering_properties}.

\subsection{Connection to Distributed Source Coding}

As the elements used for the schemes presented (quantization and binning) are closely related to distributed compression, we present here some material regarding the connection of the ensembles presented to such problems. 

A covering-good code achieves the rate-distortion function of a binary symmetric source with respect to the Hamming distortion measure, which amounts to a distortion of $\deltaGV(R)$. Furthermore, it does so with a zero error probability.

A spectrum-good code is directly applicable to the Slepian-Wolf (SW) problem~\cite{SlepianWolf73}, where the source is DSBS. Specifically, a partition of all the binary sequences into bins of rate $\Rbin$ can be performed by a code of rate $R=1-\Rbin$ (which can be alternatively be seen as a nested code with $R+\Rbin=1$), and the SW decoder can be seen as a channel decoder where the input codebook is the collection of sequences in the bin and the channel output is $Y^n$, see \cite{Wyner74}. Thus, it achieves the exponent 
\begin{align*} \BestKnownAchievableErrExp{p}{\Rbin} = \BestKnownAchievableErrExp{p}{1-R} \end{align*}
As in channel coding, this  error exponent is achievable universally in $p$. 

The achievability of the random-coding and expurgated exponents for the general discrete SW problem was established by \Csiszar{} et al.~\cite{CsiszarKorner1980} and~\cite{Csiszar1982},\footnote{Indeed, \Csiszar{} has already established the expurgated exponent for a class of additive source-pairs which includes the DSBS in~\cite{Csiszar1982}. However, as the derivation was for general rate pairs rather than for the side-information case, it faced inherent difficulty in expurgation in a distributed setting. This was solved by using linear codes; see \cite{HaimKochmanErez:2017:Full} for a detailed account in a channel-coding setting.} 
Indeed, the connection between channel coding and SW coding is fundamental (as already noted in~\cite{Wyner74}), and the optimal error exponents (if they exist) are related, see~\cite{ChenEtAl2017,Chen2007Allerton,WeinbergerMerhav2015SW}.

Nested codes are directly applicable to the Wyner-Ziv (WZ) problem~\cite{WynerZiv76}, where the source is DSBS and under the Hamming distortion measure, see \cite{ZamirShamaiErez02}. When a good ensemble is used, the exponent of a binning error event is at least $\BestKnownAchievableErrExp{p}{\Rbin}$. As the end goal of the scheme is to achieve low distortion with high probability, the designer has the freedom to choose $\Rbin$ that strikes a good balance between binning errors and other excess-distortion events, see~\cite{KellyWagner2012}.

%
\section{One-Sided Constraint: New Result}
%
\label{sec:one_user_constrained_case}

In this section we present new achievable exponent tradeoffs for the same one-sided case considered in the previous section. To that end, we will employ the same binning strategy of Corollary~\ref{cor:SHA94:DSBS}. However, our analysis technique allows to improve the exponent, and to extend it from the Stein setting to the full tradeoff. 

For our exponent region, we need the following exponent. It is a variation upon $\EBTinBallSymbol$ (Definition~\ref{def:EBTinball}), where the fixed-type noise is replaced by a noise uniform over a Hamming ball.


\begin{definition}
\label{def:EBBinball}
Fix some parameters $p,a,\thd,w \in [0,1/2]$. Let $\cc_n \in \BF^n, n=1,2,\ldots$ be a sequence of vectors such that $\lim_{n \mGoesTo \infty} \nwH{\cc_n}=w$.
Let $\ZZ \BerV{n}{p}$ and let $\NN \UCBall{n}{a}$. Define
\begin{align}
\label{eq:ebb_def}
\EBBinball{p}{a}{w}{\thd} \triangleq
\lim_{n\rightarrow\infty} - \frac{1}{n} \log \pr{\NN \oplus \ZZ \in \Ball{n}{\cc_n}{\thd}}.
\end{align}
\end{definition}

The following can be shown using standard type considerations.
\begin{lemma}
\begin{align} \label{eq:sphere_vs_ball}
\EBBinball{p}{a}{w}{\thd} = -\hB{a} + \min_{0 \leq r \leq a} \left[ \hB{r} + \EBTinBall{p}{r}{w}{\thd} \right]
\end{align}
\end{lemma}

We are now ready to state the main result of this section.

\begin{theorem}[Binary Hypothesis Testing with One-Sided Constraint]
\label{thm:SI_decision_exponents}
Consider the hypothesis testing problem as defined in Section~\ref{sec:problem_statement} for the DSBS,
with a rate constraint $R_X \in [0,1]$.  For any parameters 
$a \in [0,\deltaGV(R_X)]$ and $\thd \in [a*p_0,a*p_1]$,
\begin{align}
\left( \swEo{a,\thd}{p_0,R_X}, \swEl{a,\thd}{p_1,R_X} \right) \in \DhtAchvX{}(R_X), 
\end{align}
where
\begin{subequations}
\label{eq:DHTswE}
\begin{align}
\label{eq:DHTswEo}
&\swEo{a,\thd}{p_0,R_X} \mDefine \min \bigg\{ \EBBinball{p_0}{a}{1}{1-\thd}, \BestKnownAchievableErrExp{a*p_0}{\Rbin} \bigg\}
\\
\label{eq:DHTswEl}
&\swEl{a,\thd}{p_1,R_X} \mDefine \min \bigg\{ \EBBinball{p_1}{a}{0}{\thd}, E_c(p_1,a,\thd,\Rbin)  \bigg\},
\end{align}
\end{subequations}
where
\begin{align}
E_c(p_1,a,\thd,\Rbin) \mDefine
\max \Big\{& -\Rbin + \min_{\deltaGV(\Rbin) < w \leq 1 } \dB{w}{1/2} + \EBBinball{p_1}{a}{w}{\thd},
\BestKnownAchievableErrExp{a*p_1}{\Rbin} \Big\},
\end{align}
and where
\begin{align} \label{eq:Rbin}
\Rbin \mDefine 1 - \hB{a} - R_X
\end{align}
and $\EBBinball{p}{a}{w}{\thd}$ is defined in Definition~\ref{def:EBBinball}.
\end{theorem}

We prove this theorem using a quantize-and-bin strategy similar to that of Corollary~\ref{cor:SHA94:DSBS}, implemented with good nested codes as defined in Section~\ref{subs:nested}. In each of the two minimizations in \eqref{eq:DHTswE}, the first term is a bound on the exponent  
of a decision error resulting from a bin-decoding success, while the second is associated with a bin-decoding error, similar to the minimization in~\eqref{eq:SHAs_expression}. Using the properties of good codes, we provide a tighter and more general (not only a Stein exponent) bound; the key part is the derivation of $E_c$, the error exponent given $H_1$, and given a bin-decoding error: we replace the worst-case assumption that the ``channel output'' is uniform over all binary sequences by the true distribution, centered around the origin; in particular, it means that given an error, points close to the decision region of the correct codeword, thus not very close to any other codeword, are more likely. 

After the proof, we remark on the tightness of this result.

\begin{proof}
%
%
For a chosen $a$, denote
\begin{subequations}
\begin{align}
R_Q &\mDefine \deltaGV^{-1}(a)
\\&=1-\hB{a} \\
&= R_X + \Rbin.
\end{align}
\end{subequations}

Consider a sequence of linear nested codes $(\Cn_1,\Cn_2)$, $n=1,2,\ldots$ with rate pair $(R_Q,\Rbin)$, which is good in the sense of Definition~\ref{def:good_nested}. For convenience, the superscript of the code index in the sequence will be omitted.
The scheme we consider uses structured  quantization and binning. 
%
%
Specifically, given a vector $\XX$, we denote its quantization by \begin{align} \label{eq:defU} \UU \mDefine Q_{\sC_1}(\XX).\end{align} 
Note that we can decompose $\YY$ as follows:
\begin{subequations}
\begin{align}
\YY & = \XX \oplus \ZZ \\
& = \UU \oplus \NN \oplus \ZZ,
\end{align}
where the quantization noise $\NN = \XX \ominus \UU$ is independent of $\UU$ (and of course also of $\ZZ$) since $\XX$ is uniformly distributed. 
\end{subequations}

For sake of facilitating the analysis of the scheme, we also define
\begin{subequations}
\begin{align}
\YY' = \UU\oplus\NN' \oplus\ZZ, \label{eq:ytag}
\end{align}
where 
\begin{align}
\NN' \UCBall{n}{a} 
\label{eq:NN'}
\end{align}
\end{subequations}
is independent of the pair $(\UU,\ZZ)$. 
That is, $\YY'$ satisfies the same relations with $\UU$ as $\YY$, except that the quantization noise $\NN$ is replaced by a spherical noise with the same circumradius.

The encoded message is the syndrome increment (recall \eqref{eq:syndrome_increment}) of $\UU$, i.e.,
\begin{subequations}
\label{eq:def_phiX}
\begin{align}
\phi_X(\XX) &= \Delta \SS \\ &=  \UU \cdot \Delta H^T.
\end{align}
\end{subequations}
Since the rates of $\sC_1$ and $\sC_2$ are $R_Q$ and $\Rbin$, respectively, the encoding rate is indeed $R_X$.
%
%

Let
\begin{align*}
\SS \mDefine \UU H_2^T.
\end{align*}
By Proposition~\ref{prop:syndrome}, since $\UU\in\sC_1$, the decoder can recover from $\Delta \SS$ the coset $\sC_{2,\SS}$ of syndrome $\SS$.

The reconstructed vector at the decoder is given by \begin{align}\hat{\UU} \mDefine Q_{\sC_{2,\SS}}(\YY), \end{align} 
Denote
\begin{align}
\hat{W} &\mDefine \nwH{\YY \ominus \hat{\UU}}.
\end{align}
After computing $\hat{W}$, given a threshold $\thd\in[a*p_0,a*p_1]$, the \decider{} is given by
\begin{subequations}
\begin{align}
\psi(\phi_X(\XX),\YY) \mDefine \varphi_\thd(\hat{W}),
\end{align}
\end{subequations}
where $\varphi_\thd(w)$ is the threshold function~\eqref{eq:optimal_decision}.

%
%
Denote
\begin{align}
W &\mDefine \nwH{\YY \ominus \UU}.
\end{align}
Denote the decoding error event by $\sE_C \mDefine \{\hat{\UU}\neq\UU\}$ and the complementary event by $\mcomp{\sE_c}$. Using elementary probability laws, we can bound the error events given the two hypotheses as follows:
\begin{subequations}
\begin{align}
\epsilon_0 &= \pr{\hat{W} > \thd \mid H=0}
\\&= \pr{\sE_c, \hat W > \thd \mid H=0} + \pr{\mcomp{\sE}_c, \hat W > \thd \mid H=0}
\\&\leq \pr{\sE_c, \hat W > \thd \mid H=0} + \pr{W > \thd \mid H=0}
\\&\leq \pr{\sE_c \mid H=0} + \pr{W \geq \thd \mid H=0}. \label{eq:SI:proof:H0:last}
\end{align}
\end{subequations}

And similarly,

\begin{subequations}
\begin{align}
\epsilon_1 &= \pr{\hat{W} \leq \thd \mid H=1}
\\&= \pr{\sE_c, \hat W \leq \thd \mid H=1} + \pr{\mcomp{\sE}_c, \hat W \leq \thd \mid H=1} \\
&\leq \pr{\sE_c, \hat W \leq \thd \mid H=1} + \pr{W \leq \thd \mid H=1}. \label{eq:SI:proof:H1:last}
\end{align}
\end{subequations}

Comparing with the required exponents~\eqref{eq:DHTswE}, it suffices to show the following four exponential inequalities.
\begin{subequations}
\begin{align}
\pr{W \geq \thd \mid H=0} & \dotleq \EBBinball{p_0}{a}{1}{1-\thd} \label{eq:first} \\
\pr{\sE_c \mid H=0} & \dotleq \BestKnownAchievableErrExp{a*p_0}{\Rbin} \label{eq:second} \\
\pr{W \leq \thd \mid H=1} & \dotleq \EBBinball{p_1}{a}{0}{\thd} \label{eq:third} \\
\pr{\sE_c, \hat W \leq \thd \mid H=1} & \dotleq E_c(p_1,a,\thd,\Rbin) \label{eq:fourth}
\end{align}
\end{subequations}

In the rest of the proof we show these. For~\eqref{eq:first}, we have:
\begin{subequations}
\begin{align}
&\pr{W \geq \thd \mid H=0}
\\&= \pr{\nwH{\YY \ominus \UU} \geq \thd \mid H=0}
\\&= \pr{\NN \oplus \ZZ \notin \Ball{n}{\mvec{0}}{\thd} \mid H=0} \label{eq:SI:proof:H0:subs1}
\\&= \pr{\NN \oplus \ZZ \in \Ball{n}{\mvec{1}}{1-\thd} \mid H=0} 
\\&\dotleq \pr{\NN' \oplus \ZZ \in \Ball{n}{\mvec{1}}{1-\thd} \mid H=0} \label{eq:SI:proof:H0:nontypicality}
\\&\doteq 2^{-n \EBBinball{p_0}{a}{1}{1-\thd}},
\end{align}
\end{subequations}
where $\mvec{1}$ is the all-ones vector, \eqref{eq:SI:proof:H0:subs1} follows by substituting \eqref{eq:defU}, the transition \eqref{eq:SI:proof:H0:nontypicality} is due to Lemma~\ref{lem:quantization_uniform} and the last asymptotic equality is due to Definition~\ref{def:EBBinball}. The proof of~\eqref{eq:third} is very similar and is thus omitted.  

For~\eqref{eq:second}, we have:
\begin{subequations}
\begin{align}
&\pr{\sE_C \mid H=0}
\\&=\pr{\hat{\UU}\neq\UU \mid H=0}
\\&=\pr{Q_{\sC_{2,\SS}}(\XX \oplus \ZZ) \neq Q_{\sC_1}(\XX) \mid H=0}
\\&\dotleq \pr{Q_{\sC_{2,\SS}}\left(Q_{\sC_1}(\XX) \oplus \NN' \oplus \ZZ \right) \neq Q_{\sC_1}(\XX) \mid H=0} \label{eq:SI:proof:H0:binning_error}
\\&= \pr{Q_{\sC_{2,\SS}}\left(\UU \oplus \NN' \oplus \ZZ \right) \neq \UU \mid H=0}
\\&= \pr{Q_{\sC_{2,\SS}}(\YY') \neq \UU \mid H=0}
\\&\dotleq 2^{-n \BestKnownAchievableErrExp{a*p_0}{\Rbin}},
\end{align}
\end{subequations}
where~\eqref{eq:SI:proof:H0:binning_error} is due to Corollary~\ref{cor:nested_exponent}, the last inequality follows from the spectrum-goodness of the coarse code $\sC_{2}$  and $\YY'$ was defined in \eqref{eq:ytag}. Notice that the channel exponent is with respect to an i.i.d. noise, but it is easy to show that the exponent of a mixed noise can only be better.

%
%

Lastly, For~\eqref{eq:fourth} we have:
\begin{subequations}
\begin{align}
&\pr{\sE_C, \hat W \leq \thd \mid H=1}
\\&=\pr{\hat{\UU}\neq\UU, \hat W \leq \thd \mid H=1}
\\&=\pr{\hat{\UU}\neq\UU, \YY\in\Ball{n}{\hat{\UU}}{\thd} \mid H=1}
\\&=\pr{ \bigcup_{\cc\in\sC_{2,\SS}\setminus\{\UU\}} \left\{ \hat{\UU}=\cc,  \YY\in  \Ball{n}{\cc}{\thd} \right\} \mid H=1} 
\\&=\pr{\hat{\UU}\neq\UU, \YY\in\bigcup_{\cc\in\sC_{2,\SS}\setminus\{\UU\}} \Ball{n}{\cc}{\thd} \mid H=1} \label{eq:before_split}
\\&\dotleq \max \left\{ \pr{\sE_C \mid H=1}, \pr{\YY\in\bigcup_{\cc\in\sC_{2,\SS}\setminus\{\UU\}} \Ball{n}{\cc}{\thd} \mid H=1} \right\} . \label{eq:after_split}
\end{align}
\end{subequations}
Due to the spectrum-goodness of the coarse code $\sC_2$, the first term in the maximization is exponentially upper-bounded by 
\begin{align*} 2^{-n \BestKnownAchievableErrExp{a*p_1}{\Rbin}}. \end{align*}
For the second term, we proceed as follows.
\begin{subequations}
\begin{align}
&\pr{\YY\in\bigcup_{\cc\in\sC_{2,\SS}\setminus\{\UU\}} \Ball{n}{\cc}{\thd} \mid H=1} 
\\&\dotleq \pr{\YY'\in\bigcup_{\cc\in\sC_{2,\SS}\setminus\{\UU\}} \Ball{n}{\cc}{\thd} \mid H=1} 
\label{eq:SIthm:proof:H1:subeq:quantization_uniform}
\\&= \pr{\YY'\ominus\UU\in\bigcup_{\cc\in\sC_{2,\SS}\setminus\{\UU\}} \Ball{n}{\cc\ominus\UU}{\thd} \mid H=1} 
\\&= \pr{\YY'\ominus\UU\in\bigcup_{\cc\in\sC_2\setminus\{\mvec{0}\}} \Ball{n}{\cc}{\thd} \mid H=1} 
\\&= \pr{\NN'\oplus\ZZ\in\bigcup_{\cc\in\sC_2\setminus\{\mvec{0}\}} \Ball{n}{\cc}{\thd} \mid H=1} 
\\&\leq \sum_{\cc\in\sC_2\setminus\{\mvec{0}\}} \pr{\NN'\oplus\ZZ\in\Ball{n}{\cc}{\thd} \mid H=1}
\label{eq:union_for_uri}
\\&= \sum_{n\deltaGV(\Rbin)\leq j \leq n } \sum_{\cc\in\sC_2:n\nwH{\cc}=j} \pr{\NN' \oplus\ZZ\in\Ball{n}{\cc}{\thd} \mid H=1} \label{eq:double_sum}
\\&\doteq \sum_{n\deltaGV(\Rbin)\leq j \leq n } \sum_{\cc\in\sC_2:n\nwH{\cc}=j} 2^{-n \EBBinball{p_1}{a}{j/n}{\thd}}
\\&= \sum_{n\deltaGV(\Rbin)\leq j \leq n } \card{\sC_2} \cdot \dspectrum{\sC_\mvec{0}}{w} \cdot 2^{-n \EBBinball{p_1}{a}{j/n}{\thd}}
\\&\dotleq \sum_{n\deltaGV(\Rbin) \leq j \leq n} 2^{n\Rbin} \cdot 2^{-n \dB{j/n}{1/2}} \cdot 2^{-n \EBBinball{p_1}{a}{j/n}{\thd}} \label{eq:use_spectrum}
\\&\doteq 2^{-n \left[-\Rbin+\min_{\deltaGV(\Rbin) < w \leq 1 } \dB{w}{1/2} + \EBBinball{p_1}{a}{w}{\thd} \right]}
\label{eq:ThmSI:proof:step:Ec}
\\&= 2^{-n E_c(p_1,a,\thd,\Rbin)}
,
\end{align}
\end{subequations}
where \eqref{eq:SIthm:proof:H1:subeq:quantization_uniform} is due to Lemma~\ref{lem:quantization_uniform}, \eqref{eq:union_for_uri} is due to the union bound, the lower limit in the outer summation in \eqref{eq:double_sum} is valid since spectrum-good codes have no non-zero codewords of lower weight, in \eqref{eq:use_spectrum} we substituted the spectrum of a spectrum-good code, and in~\eqref{eq:ThmSI:proof:step:Ec} we substitute $w=j/n$ and use Laplace's method.

\end{proof}

At this point we remark on the tightness of the analysis above.

\begin{remark}
There are two points where our error-probability analysis can be improved.
\begin{enumerate}
\item For the exponent of the probability of bin-decoding error (under both hypotheses) we used $\BestKnownAchievableErrExp{a*p_i}{\Rbin}$. However, one may use the fact the quantization noise is not Bernoulli but rater bounded by a sphere to derive a larger exponent.
\item  In \eqref{eq:before_split} we have the probability of a Bernoulli noise to fall within a Hamming ball around some non-zero codeword, and also outside the basic Voronoi cell. In the transition to \eqref{eq:after_split} we bound this by the maximum between the probabilities of being in the Hamming balls and being outside the basic Voronoi cell. 
\end{enumerate}
While solving the first point is straightforward (though cumbersome), the second point (exponentially tight evaluation of \eqref{eq:before_split}) is an interesting open problem, currently under investigation. We conjecture that except for these two points, our analysis of this specific scheme is exponentially tight.  
\end{remark}

\begin{remark}
In order to see that the encoder can be improved, consider the Stein-exponent bound, obtained by setting $\thd=a*p_0$ in \eqref{eq:DHTswEl}:
\begin{align} \label{eq:new_Stein_full}
\steinElX(R_X) \geq 
 \min \bigg\{ \EBBinball{p_1}{a}{0}{a*p_0}, E_c(p_1,a,a*p_0,R_X)  \bigg\},
\end{align}
cf. the corresponding expression of the scheme by Shimokawa et al. \eqref{eq:SHAs_expression}, 
\begin{align*}
 \min \bigg\{ \EBTinBall{p_1}{a}{0}{a*p_0}, \sigma_\text{SHA}(R_X,a)\bigg\}. \end{align*} 
Now, one can show that we have an improvement of the second term. However, clearly by \eqref{eq:sphere_vs_ball}, $\EBBinball{p_1}{a}{0}{a*p_0} \leq \EBTinBall{p_1}{a}{0}{a*p_0}$. That is, quite counterintuitively, a quantization noise that has always weight $a$ is better than one that may be smaller. The reason is that a ``too good'' quantization noise may be confused by the decision function with a low crossover probability between $X$ and $Y$, favoring $\hat H = 0$. In the Stein case, where we do not care at all about the exponent of $\epsilon_0$, this is a negative effect. 
We can amend the situation by a simple tweak: the encoder will be the same, except that when it detects a quantization error that is not around $a$ it will send a special symbol forcing $\hat H = 1$. It is not difficult to verify that this will yield the more favorable bound
\begin{align} \label{eq:new_Stein}
\steinElX(R_X) \geq 
 \min \bigg\{ \EBTinBall{p_1}{a}{0}{a*p_0}, E_c(p_1,a,a*p_0,R_X)  \bigg\}.
\end{align}
A similar process, where if the quantization noise is below some threshold $\hat H=1$ is declared, may also somewhat extend the exponent region in the regime ``close'' to Stein (low $E_0$), but we do not pursue this direction.
\end{remark}

\begin{remark} Of course, the two-stage decision process where $\hat H$ is a function of $\hat W$ is sub-optimal. It differs from the Neyman-Pearson test that takes into account the probability of all possible values of $W$.
\end{remark}

\begin{remark}
Using time sharing on top of the scheme, we can obtain improved performance according to Proposition~\ref{prop:time-sharing:general}.
\end{remark}

\begin{remark} 
In the special case $a=0$ the scheme amounts to binning without quantization, and the nested code maybe replaced by a single spectrum-good code. In this case the expressions simplify considerably, and we have the pair:
 \begin{subequations}
\label{eq:DHTswE_0}
\begin{align}
\label{eq:DHTswEo_0}
&\swEo{0,\thd}{R_X} = \min \bigg\{ \dB{\thd}{p_0}, \BestKnownAchievableErrExp{p_0}{\Rbin} \bigg\}
\\
\label{eq:DHTswEl_0}
&\swEl{0,\thd}{R_X} \mDefine \min \bigg\{ \dB{\thd}{p_1} , E_c(p_1,0,\thd,\Rbin)  \bigg\},
\end{align}
\end{subequations}
where
\begin{align}
E_c(p_1,0,\thd,\Rbin) =
\max \Big\{& -\Rbin + \min_{\deltaGV(R_X) < w \leq 1 } \dB{w}{1/2} + \EBBinball{p_1}{0}{w}{\thd}, 
\BestKnownAchievableErrExp{p_1}{\Rbin} \Big\},
\end{align}
and where $\EBBinball{p_1}{0}{w}{\thd}$ is defined in \eqref{eq:ebb_def}.
\end{remark}

%
\section{Symmetric Constraint}
%
\label{sec:symmetric_rate_constraint}

In this section we proceed to a symmetric rate constraint $R_X=R_Y=R$. 
In this part our analysis specifically hinges on the linearity of codes, and specifically builds on the \KM{} coding scheme~\cite{KornerMarton79}. 
Adding this ingredient to the analysis, we get an achievable exponent region for the symmetric constraint in the same spirit of the achievable region in Theorem~\ref{thm:SI_decision_exponents}, where the only loss due to constraining $R_Y$ is an additional spherical noise component. Before stating our result, we give some background on the new ingredient.

%
\subsection{\KM{} Compression}
%
\label{sec:km_background}

The \KM{} problem has the same structure as our DHT problem for the DSBS, except that the crossover probability is known (say $p$), and the decision function is replaced by a decoder, whose goal is to reproduce the difference sequence $\ZZ = \YY \ominus \XX$ with high probability. By considering the two corresponding one-sided constrained problems, which amount to SI versions of the SW problem, it is clear that any rate $R<\hB{p}$ is not achievable. The \KM{} scheme allows to achieve any rate $R>\hB{p}$ in the following manner.

Assume that the two encoders use the \emph{same} linear codebook, with parity-check matrix $H$. Further, both of them send the syndrome of their observed sequence:
\begin{subequations}
\label{eq:KM_encoders}
\begin{align}
&\phi_X(\XX) = \XX \mat{H}^T
\\
&\phi_Y(\YY) = \YY \mat{H}^T.
\end{align}
\end{subequations}
In the first stage of the decoder, the two encoded vectors are summed up, leading to
\begin{subequations}
\label{eq:KM_decoder}
\begin{align}
\phi_Y(\YY) \ominus \phi_X(\XX) &= \YY \mat{H}^T \ominus \XX \mat{H}^T
\\&= \ZZ \mat{H}^T,
\end{align}
\end{subequations}
which is but the syndrome of the difference sequence. This is indistinguishable from the situation of a decoder for the SI SW problem, 

The decoder is now in the exact same situation as an
optimal decoder for a BSC with crossover probability $p$, and code rate $1-R$. By the fact that linear codes allow to approach the capacity $1-H_b(p)$, the optimal rate follows. Further, if spectrum-good codes are used, the exponent $\BestKnownAchievableErrExp{p}{1-R}$ is achievable, i.e., there is no loss in the exponent w.r.t. the corresponding side-information SW problem.

\subsection{A New Bound}

We now present an achievability result that relies upon a very simple principle:as in the \KM{} decoder, after performing the XOR \eqref{eq:KM_decoder} the situation is indistinguishable from that of the input to a SW decoder, also in DHT we can 

\begin{theorem}[Binary Hypothesis Testing with Symmetric Constraint]
\label{thm:KM_decision_exponents}
Consider the hypothesis testing problem as defined in Section~\ref{sec:problem_statement} for the DSBS,
with a symmetric rate constraint $R \in [0,1]$.  For any parameter
$\thd \in [p_0,p_1]$,
\begin{align}
\left( \kmEo{\thd}{R}, \kmEl{\thd}{R} \right) \in \DhtAchv{}(R), 
\end{align}
where 
\begin{subequations}
\label{eq:DHTkmE}
\begin{align}
\label{eq:DHTkmEo}
\kmEo{\thd}{R} & = \swEo{0,\thd}{R} 
\\
\label{eq:DHTkmEl}
\kmEl{\thd}{R} & = \swEl{0,\thd}{R} ,
\end{align}
\end{subequations}
where the one-sided constraint exponents with $a=0$ are given in \eqref{eq:DHTswE_0}.
\end{theorem}

\begin{proof}
Let the codebook be taken from a spectrum-good sequence. Let the encoders be the \KM{} encoders \eqref{eq:KM_encoders}. The decision function first obtains $\ZZ H^T$ as in the \KM{} decoder \eqref{eq:KM_decoder}, and then evaluates
\[ \hat\ZZ = Q_{C} (\ZZ H^T) \]
and applies the threshold function to
\[ W' \mDefine \nwH{\hat \ZZ}. \]

Noticing that $W'$ is equal in distribution to $\hat W$  in the proof of Theorem~\ref{thm:SI_decision_exponents} when $a=0$, and that all error events only functions of that variable, the proof is completed.
\end{proof}

It is natural to ask, why we restrict ourselves under a symmetric rate constraint to a binning-only scheme. Indeed, lossy versions of the \KM{} problem have been studied in~\cite{KrithivasanPradhan2011,Wagner2011}. One can construct a scheme based on nested codes, obtain a lossy reconstruction of the noise sequence $\ZZ$ and then test its weight. However, unlike the reconstruction in the single-sided case which includes a Bernoulli component ($\ZZ$) and a quantization noise bounded by a Hamming ball ($\NN$), in the symmetric-rate case we will have a combination of a Bernoulli component with \emph{two} quantization noises. This makes the analysis considerably more involved. An idea that comes to mind, is to obtain a bound by replacing at least one of the quantization noises with a Bernoulli one; however, we do not see a clear way to do that. Thus, improving Theorem~\ref{thm:KM_decision_exponents} by introducing quantization is left for future research.

\section{Performance Comparison}
\label{sec:performance_comparison}

In this section we provide a numerical comparison of the different bounds for the DSBS.

We start with the Stein setting, where we can compare with the previously-known results. Our achievable performance for the one-sided constrained case is given by \eqref{eq:new_Stein} (which coincides with \eqref{eq:new_Stein_full} for the parameters we checked). We compare it against the unconstrained performance, and against the previously best-known achievable exponent, namely the maximum between Corollaries~\ref{cor:Han87:binary_source} and \ref{cor:SHA94:DSBS}.\footnote{In~\cite{Shimokawa:1994:MScThesis}, the performance is evaluated using an asymmetric choice of the auxiliary variable. We have verified that the symmetric choice we use performs better.} 
To both, we also apply time sharing as in Proposition~\ref{prop:time-sharing:general}. It can be seen that the new exponent is at least as good, with slight improvement for some parameters. As reference, we show the unconstrained performance, given by Corollary~\ref{cor:unconstrained:DSBS}. Also shown on the plots, is the performance obtained under a symmetric rate constraint, found by constraining the quantization parameter to $a=0$. It can be seen that for low $p_1$ the symmetric constraint yields no loss with respect to the one-sided constraint.

\begin{figure}[htb]
\centering
\subfigure[$p_1=0.25$]
{\includegraphics[scale=0.45]{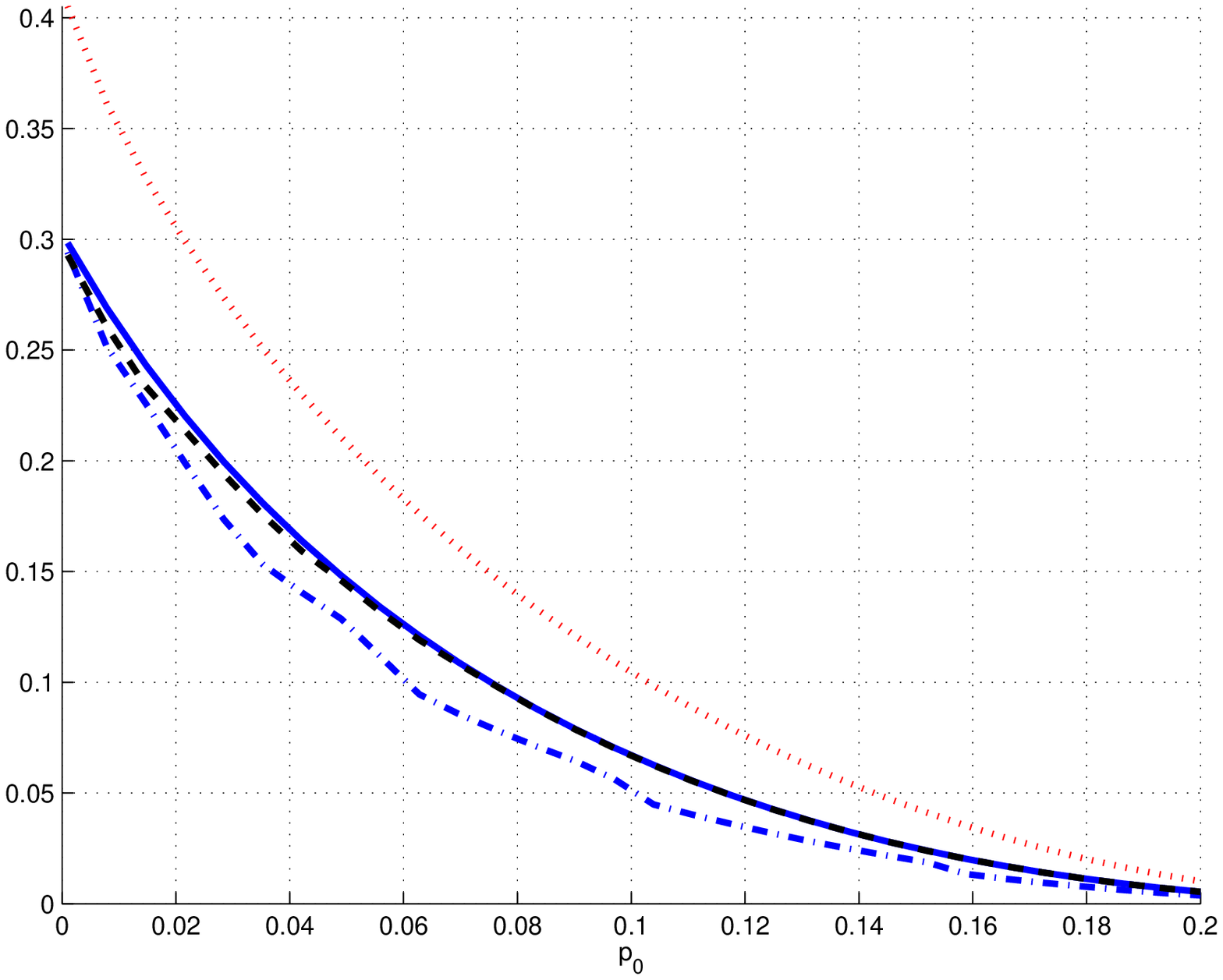}
\label{fig:performance_Stein_0.1}}
\subfigure[$p_1=0.1$. The three top curves coincide.]
{\includegraphics[scale=0.45]{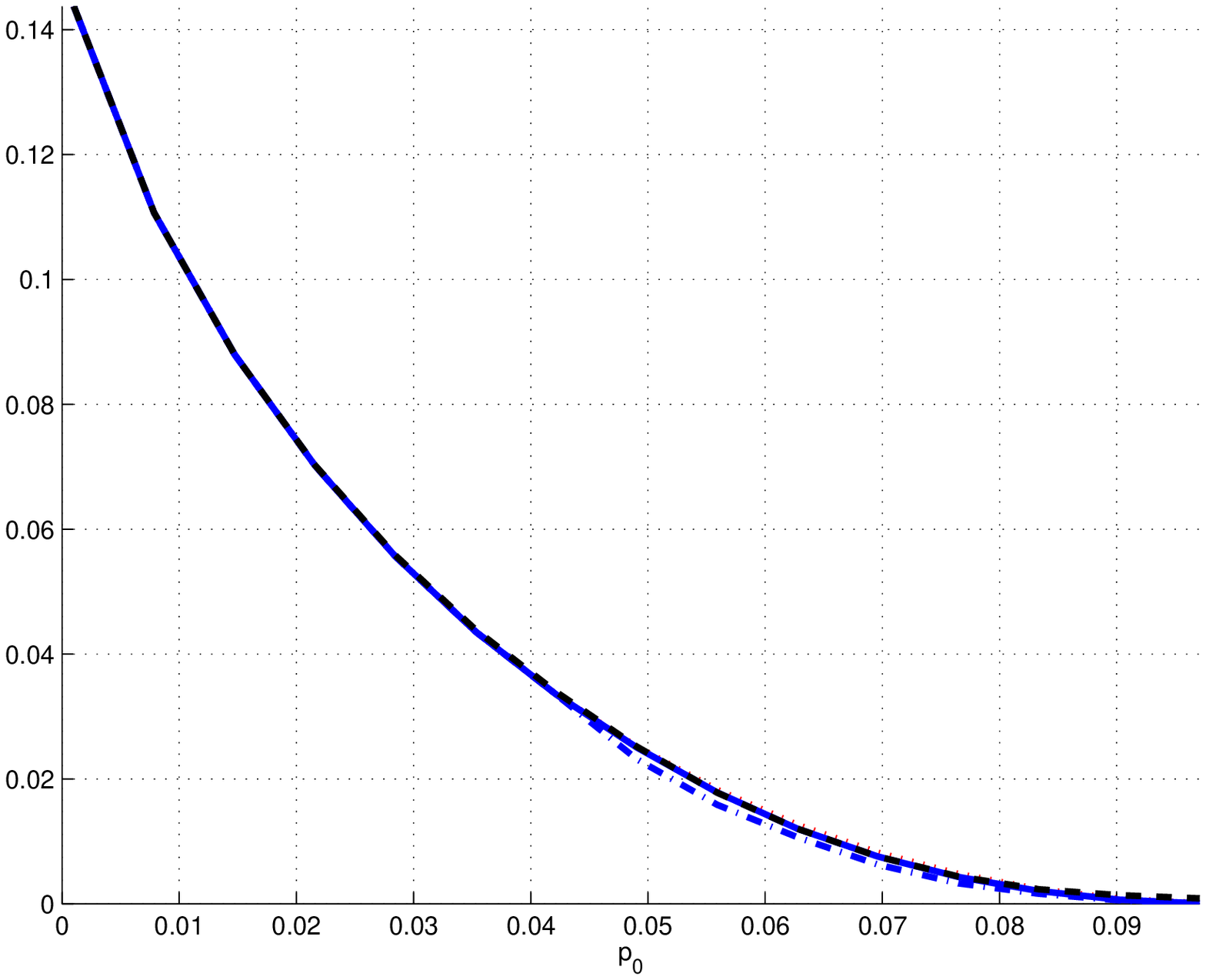}
\label{fig:performance_Stein_0.25}}
\caption{Stein exponent comparison. The rate is $0.3$ bits, the X axis is $p_0$. From top to bottom: unconstrained performance, new exponent, previously known exponent, new exponent without quantization (achievable with symmetric rate).}
\end{figure}

Beyond the Stein setting, we plot the full exponent tradeoff, achievable by Theorems~\ref{thm:SI_decision_exponents} and~\ref{thm:KM_decision_exponents}. In this case we are not aware of any previous results we can compare to. We thus only add the unconstrained tradeoff of Corollary~\ref{cor:unconstrained:DSBS}, and the simple strategy of Corollary~\ref{cor:straightforward_approach:DSBS}. Also here it can be seen that the symmetric constraint imposes a loss for high $p_1$, but not for a lower one.

\begin{figure}[htb]
\centering
\subfigure[$p_1=0.25$]
{\includegraphics[scale=0.45]{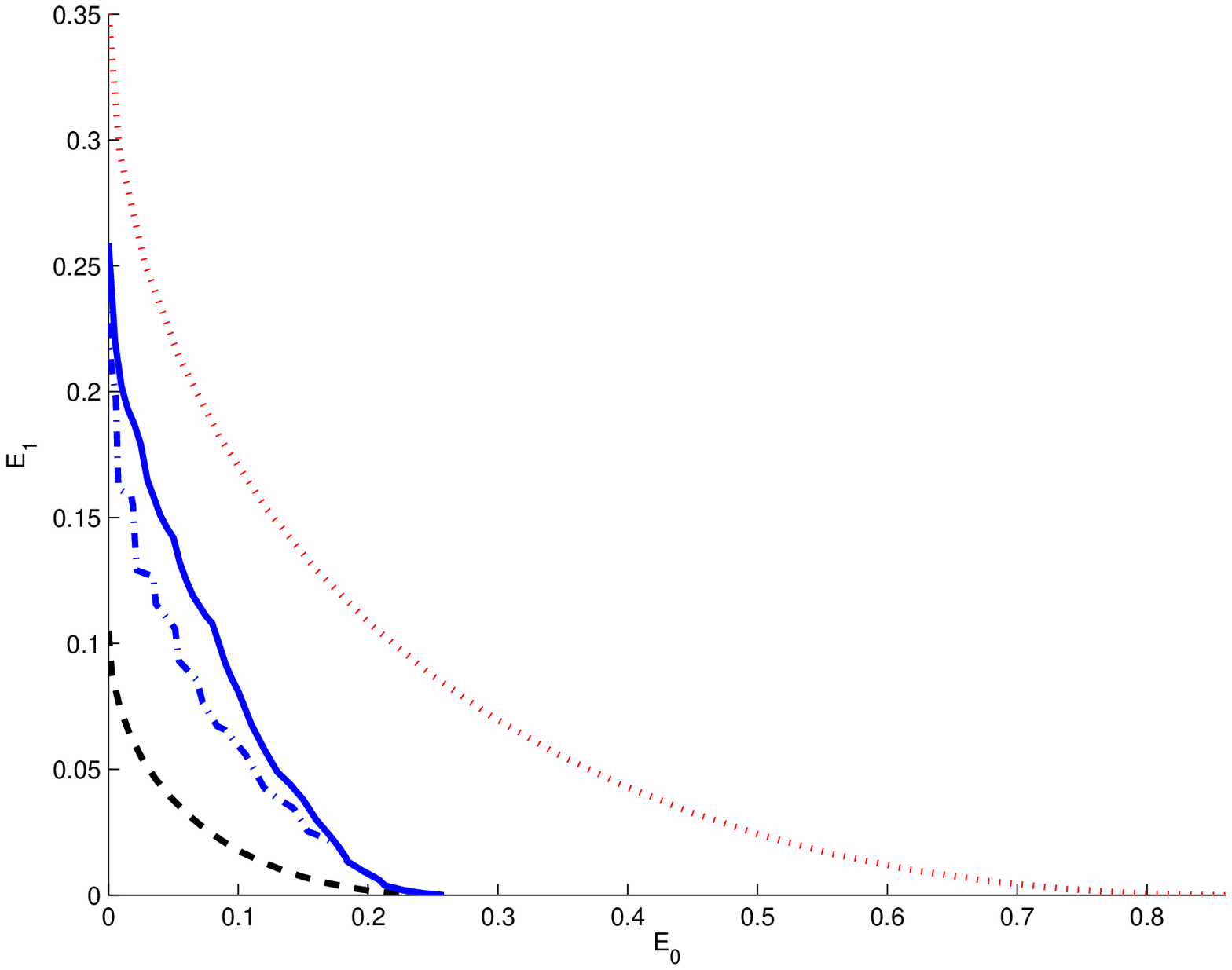}
\label{fig:performance_tradeoff_0.1}}
\subfigure[$p_1=0.1$. The new tradeoff curves with and without quantization coincide.]
{\includegraphics[scale=0.45]{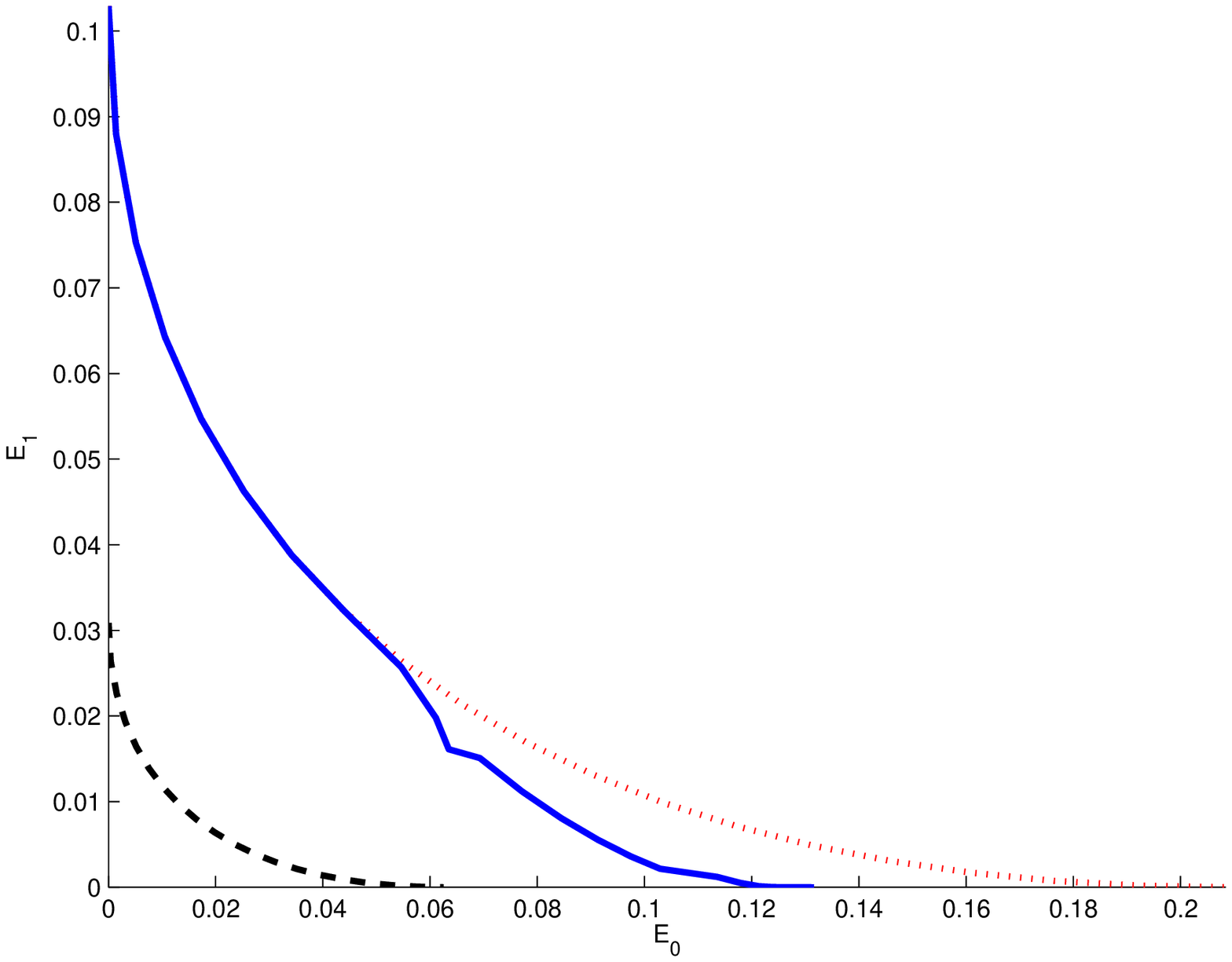}
\label{fig:performance_tradeoff_0.25}}
\caption{Exponent tradeoff comparison. The rate is $0.3$ bits, $p_0=0.01$. From top to bottom: unconstrained tradeoff, new tradeoff, new tradeoff without quantization (achievable with symmetric rate), time-sharing only tradeoff.}
\end{figure}

\section{Conclusions}
\label{sec:future_work}
%
In this work we introduced new achievable error exponents for binary distributed hypothesis testing for binary symmetric i.i.d sources. One may 
wonder, naturally, regarding the extension beyond the binary symmetric case.

In that respect, a distinction should be made between two parts of the work. Under a one-sided rate constraint, linear codes were used merely for
concreteness and for convenience of extension to the symmetric constraint; the same results should hold for a random code ensemble. Thus, there is no fundamental problem in extending our
analysis to any discrete memoryless model. 

In contrast, in the setting of a symmetric rate constraint, we explicitly use the ``matching'' between the closedness under addition of linear codes, and the additivity of the relation between the sources. Thus, our approach which has almost no loss with respect to the single-sided constraint, cannot 
be extended beyond (nearly) additive cases. The question whether a different approach can achieve that, remains open.

Finally, we stress again the lack of tight outer bounds for the problem, except for cases where the communication  constraints do not limit the exponents.

%
\section*{Acknowledgment}
%
The authors thank Uri Erez for sharing insights throughout the work. They also thank Vincent Y. F. Tan for introducing them the distributed hypothesis testing problem, and Nir Weinberger for helpful discussions.

%
\appendices
%
\newcommand{\Types}{\mathcal{T}}
\newcommand{\B}{\mathcal{B}}

\appendices

\section{Exponent of a Hamming Ball}
\label{sec:exp_hamming_ball}

In this appendix we evaluate the exponent of the event of a mixed noise entering a Hamming ball, namely $\EBTinBall{p}{a}{w}{\thd}$ of Definition~\ref{def:EBTinball}.

\begin{lemma}
\label{lem:EBTinBall_derivation}
\begin{align}
\EBTinBall{p}{a}{w}{\thd} = \min_{\gamma \in \left[\max(0,a+w-1),\min(w,a)\right]} &\hB{a} -w\hB{\frac{\gamma}{w}} - (1-w)\hB{\frac{a-\gamma}{1-w}} \NL{}{}{} + E_\text{w}(p, 1-(w+a-2\gamma), w+a-2\gamma, \thd-(w+a-2\gamma))
\end{align}
where
\begin{align} \label{eq:E_bb}
E_\text{w}(p,\alpha,\beta,\thd) \mDefine &\min_{x\in[\max(0,\thd),\min(\alpha,\beta+\thd)]} \alpha \dB{\frac{x}{\alpha}}{p} + \beta \dB{\frac{x-\thd}{\beta}}{p}
\end{align}
\end{lemma}

The proof follows from the lemmas below.

\begin{lemma}[Difference of weights]
Let $\ZZ_1$ and $\ZZ_2$ be two random vectors. Assume that $\ZZ_i \BerV{k_i}{p}$ and let $W_i = \wH{\ZZ_i}$ for $i=1,2$.
If $k_i$ grow with $n$ such that
\begin{align*}
    \lim_{n\rightarrow\infty} \frac{k_1}{n} &= \alpha, \\
    \lim_{n\rightarrow\infty} \frac{k_2}{n} &= \beta, \\
\end{align*}
and further let a sequence $t$ grow with $n$ such that
  \[  \lim_{n\rightarrow\infty} \frac{t}{n} = \tau \]
where $\tau \in (-\beta,\alpha)$.
Then, \[ \lim_{n\mGoesTo\infty} - \frac{1}{n} \log \pr{W_1 - W_2 = t} = E_\text{bb}(p,\alpha,\beta,\tau), \]
where
$E_\text{w}(p,\alpha,\beta,\tau)$ is given by \eqref{eq:E_bb}.
\begin{proof}
\begin{subequations}
\begin{align}
\pr{W_1 - W_2 = t} &= \sum_{w=0}^{k_1} \pr{W_1=w} \pr{W_2 = w-t}
\\&= \sum_{w=\max(0,t)}^{\min(k_1,k_2+t)} \pr{W_1=w} \pr{W_2 = w-t}
\\ \label{eq:justify} &\doteq \sum_{w=\max(0,t)}^{\min(k_1,k_2+t)} 2^{-k_1 \dB{\frac{w}{k_1}}{p}} 2^{-k_2 \dB{\frac{w-t}{k_2}}{p}}
\\&\doteq \max_{w \in [\max(0,t),\ldots,\min(k_1,k_2+t)]} 2^{-k_1 \dB{\frac{w}{k_1}}{p}} 2^{-k_2 \dB{\frac{w-t}{k_2}}{p}}
\\&\doteq \max_{x\in[\max(0,\tau),\min(\alpha,\beta+\tau)]} 2^{-n \left[ \alpha \dB{\frac{x}{\alpha}}{p} + \beta \dB{\frac{x-\tau}{\beta}}{p} \right]},
\end{align}
\end{subequations}
where \eqref{eq:justify} follows by the exponent of the probability of a type class.
\end{proof}
\end{lemma}

The following lemma will assist in proving Lemma~\ref{lem:asymp} which follows.
\newcommand{\sI}{\ensuremath{{\mset{I}}}}
\newcommand{\sM}{\ensuremath{{\mset{M}}}}
\newcommand{\bsI}{\ensuremath{{\overline{\sI}}}}
\newcommand{\bsM}{\ensuremath{{\overline{\sM}}}}
\begin{lemma}[Mixed noise, fixed dimension]\label{lem:fixed_dim}
Let $n \in \mNN$, and consider a noise that is a mixture of a noise $\UU \Uniform{\Type{n}{a}}$ (uniform over a fixed type) and a Bernoulli vector $\ZZ \BerV{n}{p}$, where $a \in \nicefrac{1}{n}\cdot\{0,\ldots,n\}$ and $p \in [0,1]$. Further let 
$\cc \in \BF^n$  where $w=\nwH{\cc} \in \nicefrac{1}{n}\cdot\{0,\ldots,n\}$, be the center of a ``distant'' sphere.
Then, for a sphere radius with normalized radius $\thd \in \nicefrac{1}{n}\cdot\{0,\ldots,n\}$,
\begin{align}
\pr{\wH{\cc \oplus \UU \oplus \ZZ} = n\thd} = \sum_{m=n\cdot\max(0,a+w-1)}^{n\cdot\min(w,a)} \frac{\binom{nw}{m}\binom{n-nw}{na-m}}{\binom{n}{na}} \pr{ W_{1,m} - W_{2,m} = n\thd-nw-na+2m }
\end{align}
where $W_{1,m} \Binomial{n-(nw+na-2m)}{p}$ and $W_{2,m} \Binomial{nw+na-2m}{p}$.
\end{lemma}
\begin{proof}
Define the following sets of indices $\sI,\sM_1,\sM_2$:
\begin{subequations}
\begin{align}
\sI &\mDefine \left\{ i : c_i = 1 \right\}
\\
\sM_1 &\mDefine \left\{ i : U_{\sI,i} = 1 \right\}
\\
\sM_2 &\mDefine \left\{ i : U_{\bsI,i} = 1 \right\}.
\end{align}
\end{subequations}
In words, $\sI$ is the set of indices where $\mathbf{c}$ contains ones;  
$\sM_1$ is the subset (within $\sI$) where $\UU$ contains ones and $\sM_2$
is defined similarly over the complement of $\sI$.

Then,
\begin{subequations}
\begin{align}
&\pr{ \wH{\cc \oplus \UU \oplus \ZZ = n\thd } }
= \pr{ \wH{\cc_{\sI}\oplus \UU_{\sI} \oplus \ZZ_{\sI}} + \wH{\cc_{\bsI} \oplus \UU_{\bsI} \oplus \ZZ_{\bsI}} = n\thd }
\\&= \pr{ nw - \wH{\UU_{\sI} \oplus \ZZ_{\sI}} + \wH{\UU_{\bsI} \oplus \ZZ_{\bsI}} = n\thd }
\\&= \pr{ nw - \left[ \wH{\UU_{\sI,\sM_1} \oplus \ZZ_{\sI,\sM_1}} + \wH{\UU_{\sI,\bsM_1} \oplus \ZZ_{\sI,\bsM_1}} \right] \NLa{=} + \left[ \wH{\UU_{\bsI,\sM_2} \oplus \ZZ_{\bsI,\sM_2}} + \wH{\UU_{\bsI,\bsM_2} \oplus \ZZ_{\bsI,\bsM_2}} \right] = n\thd }
\\&= \pr{ nw - \left[ M_1 - \wH{\ZZ_{\sI,\sM_1}} + \wH{\ZZ_{\sI,\bsM_1}} \right] + \left[ M_2 - \wH{\ZZ_{\bsI,\sM_2}} + \wH{\ZZ_{\bsI,\bsM_2}} \right] = n\thd }
\\&= \pr{ \wH{\ZZ_{\sI,\sM_1}} - \wH{\ZZ_{\sI,\bsM_1}} - \wH{\ZZ_{\bsI,\sM_2}} + \wH{\ZZ_{\bsI,\bsM_2}} = n(\thd-w)+M_1-M_2 }
\\&= \pr{ \wH{\ZZ_{\sI,\sM_1}} - \wH{\ZZ_{\sI,\bsM_1}} - \wH{\ZZ_{\bsI,\sM_2}} + \wH{\ZZ_{\bsI,\bsM_2}} = n(\thd-w-a)+2M_1 } \label{lemma:mixed_noise_fixed_dimension:proof:set_cardinality}
\\&= \sum_{m=n\cdot\max(0,a+w-1)}^{n\cdot\min(w,a)} \pr{M_1=m} \NL{}{=}{\cdot} \pr{ \wH{\ZZ_{\sI,\sM_1}} - \wH{\ZZ_{\sI,\bsM_1}} - \wH{\ZZ_{\bsI,\sM_2}} + \wH{\ZZ_{\bsI,\bsM_2}} = n(\thd-w-a)+2m  \mid M_1=m } \label{lemma:mixed_noise_fixed_dimension:proof:set_cardinality_calc}
\\&= \sum_{m=n\cdot\max(0,a+w-1)}^{n\cdot\min(w,a)} \frac{\binom{nw}{m}\binom{n-nw}{na-m}}{\binom{n}{na}} \pr{ W_{1,m} - W_{2,m} = n(\thd-w-a)+2m } \label{lemma:mixed_noise_fixed_dimension:proof:binomial_RVs}
\end{align}
\end{subequations}
where~\eqref{lemma:mixed_noise_fixed_dimension:proof:set_cardinality} follows since $\card{\sI}=nw$, by denoting $M_1 \mDefine \card{\sM_1}$, $M_2 \mDefine \card{\sM_2}$ and noting that $M_1+M_2 = na$;
equality~\eqref{lemma:mixed_noise_fixed_dimension:proof:set_cardinality_calc} follows since $M_1 \leq na$ and $M_1 \leq nw$, and since $M_2 \leq na$ and $M_2 \leq n(1-w)$ (therefore $M_1 \geq n(a+w-1)$); and 
equality~\eqref{lemma:mixed_noise_fixed_dimension:proof:binomial_RVs} follows by denoting $W_{1,m} \Binomial{n-(nw+na-2m)}{p}$ and $W_{2,m} \Binomial{nw+na-2m}{p}$.
\end{proof}

\begin{lemma}[Mixed noise, asymptotic dimension]
\label{lem:asymp}
Consider a sequence of problems as in Lemma~\ref{lem:fixed_dim} indexed by the blocklebgth $n$, with parameters $a_n\rightarrow a$, 
$w_n\rightarrow w$ and $\thd_n \rightarrow \thd$. 
Then:
\begin{align}
\lim_{n\mGoesTo\infty} - \frac{1}{n} \pr{\wH{\cc_n \oplus \UU_n \oplus \ZZ_n} = n\thd_n} = \EBTinBall{p}{a}{w}{\thd}.
\end{align}
\begin{proof}
A straightforward calculation shows that
\begin{subequations}
\begin{align}
&\pr{ \wH{\cc_n \oplus \UU_n \oplus \ZZ_n} = n\thd_n }
\\&\doteq \sum_{m=n\cdot\max(0,a_n+w_n-1)}^{n\cdot\min(w_n,a_n)} 2^{nw_n \hB{\frac{m/n}{w_n}}} 2^{n(1-w_n) \hB{\frac{a_n-m/n}{1-w_n}}} 2^{-n \hB{a_n}}
\NL{}{\sum_{m=n\cdot\max(0,a_n+w_n-1)}^{n\cdot\min(w_n,a_n)}}{}
\cdot 2^{-n E_\text{w}(p,1-(w_n+a_n-2m/n),w_n+a_n-2m/n,t_n-(w_n+a_n-2m/n))}
\\&\doteq \max_{m \in \left[\max(0,a+w-1),\min(w,a)\right]} 2^{-n \left[ -w\hB{\frac{m}{w}} - (1-w)\hB{\frac{a-m}{1-w}} + \hB{a} \right]} \cdot 2^{-n E_\text{w}(p, 1-(w+a-2m), w+a-2m, \thd-(w+a-2m))}.
\end{align}
\end{subequations}
\end{proof}
\end{lemma}

The proof of Lemma~\ref{lem:EBTinBall_derivation} now follows:
\begin{subequations}
\begin{align}
&\pr{ \wH{\cc_n \oplus \UU_n \oplus \ZZ_n} \leq \thd_n }
=\sum_{\tau = 0}^{\thd_n} \pr{ \wH{\cc_n \oplus \UU_n \oplus \ZZ_n} = n\tau }
\\&\doteq \sum_{\tau = 0}^{n\thd_n} 2^{-n \EBTinBall{p}{a}{w}{\tau/n} }
\\&\doteq \max_{\tau \in [0,\thd]} 2^{-n \EBTinBall{p}{a}{w}{\tau}}
\\&=  2^{-n \min_{\tau \in [0,\thd]} \EBTinBall{p}{a}{w}{\tau}}
\end{align}
\end{subequations}

\section{Quantization-Noise Properties of Good Codes}
\label{app:covering_properties}

In this section we prove Lemma~\ref{lem:quantization_uniform} and Corollary~\ref{cor:nested_exponent}, which contain the properties of good codes that we need for deriving our achievable exponents.

First we define the \emph{covering efficiency} of a code $\sC$ as: 
\begin{align}
\eta(\sC) \mDefine \frac{\card{\Ball{n}{\mvec{0}}{\NrCoveringRadius(\sC)}}}{\card{\OmegaO}}.
\end{align}

\begin{lemma}
\label{lem:covering_efficiency}
Consider a covering-good sequence of codes \SequenceOfCodes{} of rate $R$.
Then,
\begin{align}
\eta(\Cn) \doteq 1
\end{align}
\end{lemma}
\begin{proof}
Since for all $\cc \in \Cn$ we have that $\card{\Omega_\cc} = \card{\OmegaO}$, it follows that
\begin{align}
\card{\OmegaO} \doteq 2^{n(1-R)}.
\end{align}
Therefore,
\begin{subequations}
\begin{align}
\eta(\Cn) &\doteq\frac{2^{-n(1-R)}}{\card{\Ball{n}{\mvec{0}}{\NrCoveringRadius(\Cn)}}^{-1}}
\\&\doteq\frac{2^{-n(1-R)}}{2^{-n \hB{\NrCoveringRadius(\Cn)}}}
\\&= 2^{n \pb{\hB{\NrCoveringRadius(\Cn)} - (1-R)}}
\\&= 2^{n \pb{\hB{\NrCoveringRadius(\Cn)} - \hB{\deltaGV(R)}}}
\\&\doteq 1,
\end{align}
\end{subequations}
where the last asymptotic equality is due to~\eqref{eq:def:covering_good} and due to the continuity of the entropy.
\end{proof}

\begin{proof}[Proof of Lemma~\ref{lem:quantization_uniform}]
Since for any $\vv \notin \OmegaO$
\begin{subequations}
\label{proof:cor:QvsB:outside}
\begin{align}
\pr{\XX \ominus Q_{\Cn}(\XX)=\vv} &= 0
\\&\leq \pr{\NN=\vv},
\end{align}
\end{subequations}
it is left to consider points $\vv \in \OmegaO$. To that end, since the code is linear, due to symmetry
\begin{align}
\XX \ominus Q_{\Cn}(\XX) \Uniform{\OmegaO}.
\end{align}
Thus, for $\vv \in \OmegaO$, 
\begin{subequations}
\label{proof:cor:QvsB:inside}
\begin{align}
&\frac{\pr{\XX \ominus Q_{\Cn}(\XX)=\vv}}{\pr{\NN=\vv}}
\\&=\frac{\card{\OmegaO}^{-1}}{\card{\Ball{n}{\mvec{0}}{\NrCoveringRadius(\Cn)}}^{-1}}
\\&=\eta(\Cn)
\\&\doteq 1,
\end{align}
where the last (asymptotic) equality is due to Lemma~\ref{lem:covering_efficiency}. 
This completes the last part (Note that the rate of convergence is independent of $\vv$, and therefore the convergence is uniform over~$\vv\in\OmegaO$).
\end{subequations}
The second part follows by the linearity of convolution: 
For any $\vv\in\BF^n$,
\begin{subequations}
\begin{align}
&\pr{\XX \ominus Q_{\Cn}(\XX) \oplus \ZZ = \vv}
\\&= \sum_{\zz\in\BF^n} P_\ZZ(\zz) \pr{\XX \ominus Q_{\Cn}(\XX) \oplus \zz=\vv}
\\&= \sum_{\zz\in\BF^n} P_\ZZ(\zz) \pr{\XX \ominus Q_{\Cn}(\XX) = \vv \ominus \zz}
\\& \dotleq \sum_{\zz\in\BF^n} P_\ZZ(\zz) \pr{\NN = \vv \ominus \zz}
\\& = \pr{\NN \oplus \ZZ = \vv},
\end{align}
\end{subequations}
where the (asymptotic) inequality is due to the first part of the lemma.
\end{proof}

\begin{proof}[Proof of Corollary~\ref{cor:nested_exponent}]
Let $\NN=\XX \ominus \UU$. Note that since the code is linear and $\XX$ is uniform  over $\{0,1\}^n$, it follows that $\NN$ is independent of the pair $(\UU,\ZZ)$. Thus,
\begin{subequations}
\begin{align}
&\pr{Q_{\Cn_{2,\SS}}(\XX \oplus \ZZ) \neq \UU}
\\&=
\pr{Q_{\Cn_{2,\SS}}\left(\UU \oplus [\XX \ominus \UU] \oplus \ZZ\right) \neq \UU}
 \\&=
\pr{Q_{\Cn_{2,\SS}}\left(\UU \oplus \NN \oplus \ZZ\right) \neq \UU}.
\end{align}
\end{subequations}
Recalling the definition of $\NN'$ in (\ref{eq:NN'}), we have
\begin{subequations}
\begin{align}
&\pr{Q_{\Cn_{2,\SS}}\left(\UU \oplus \NN \oplus \ZZ\right) \neq \UU}
 \\&=
\pr{Q_{\Cn_{2,\SS}}\left(\UU \oplus \NN \oplus \ZZ\right) \neq \UU}
\\&=
\mathbb{E}_{\UU} \pr{Q_{\Cn_{2,\SS}}\left(\uu \oplus \NN \oplus \ZZ\right) \neq \uu \mid \UU=\uu}
\\&\dotleq
\mathbb{E}_{\UU} \pr{Q_{\Cn_{2,\SS}}\left(\uu \oplus \NN' \oplus \ZZ\right) \neq \uu \mid \UU=\uu}
\label{eq:asymp_eq}
\\&=
\pr{Q_{\Cn_{2,\SS}}\left(\UU \oplus \NN' \oplus \ZZ \right) \neq \UU},
\end{align}
\end{subequations}
where (\ref{eq:asymp_eq}) follows by applying Lemma~\ref{lem:quantization_uniform} for each $\uu$.
\end{proof}

\section{Existence of Good Nested Codes}
\label{app:good}

In this appendix we prove the existence of a sequence of good nested codes, as defined in Definition~\ref{def:good_nested}. To that end, we first state known results on the existence of spectrum-good and covering-good codes.

By \cite{GallagerBook1968}, random linear codes are spectrum-good with high probability. That is, let $\mathcal{C}^{(n)}$ be a linear code of blocklength $n$ and rate $R$, with a generating matrix $\mat{G}^{(n)}$ drawn i.i.d. Bernoulli-$1/2$. Then there exist some sequences $\epsilon_S^{(n)}$ and $\delta_S^{(n)}$ approaching zero, such that for all $w$,
\begin{align} \label{eq:complicated_spectrum}
\pr { \dspectrum{\Cn}{w} > (1+\delta_S^{(n)}) \dspectrumgood{n}{R}{w} } \leq \epsilon_S^{(n)}.
\end{align}

As for covering-good codes, a construction based upon random linear codes is given in~\cite{LitsynCoveringCodesBook1997}. For a generating matrix $\mat{G}^{(n)}$ at blocklength $n$, a procedure is given to generate a new matrix $\mat{G}_I^{(n)}=\mat{G}_I^{(n)}(\mat{G}^{(n)})$, with $k_I^{(n)} \mDefine \ceil{\log n}$ rows. The matrices are combined in the following way:
\begin{align}
\mat{G}'^{(n)} =
\left[
\begin{array}{c}
\mat{G}_I^{(n)}
\\ \hdashline
\mat{G}^{(n)}
\end{array}
\right]
\left.
\begin{array}{l}
\big\} k_I^{(n)} \times n
\\
\big\} k^{(n)} \times n
\end{array}
\right\} k'^{(n)} \times n
\end{align}
Clearly, adding $\mat{G}_I$ does not effect the rate of a sequence of codes. Let $\Cn$ be constructed by this procedure, with $\mat{G}^{(n)}$ drawn i.i.,d. Bernoulli-$1/2$. It is shown in~\cite[Theorem 12.3.5]{LitsynCoveringCodesBook1997}, that there exists sequences $\epsilon_C^{(n)}$ and $\delta_C^{(n)}$ approaching zero, such that
\begin{align} \label{eq:complicated_covering}
\pr { \NrCoveringRadius(\Cn) >
(1+\delta_C^{(n)}) \deltaGV(R)} \leq \epsilon_C^{(n)}. \end{align}

A nested code of blocklength $n$ and rates $(R,\Rbin)$ has a generating matrix
\begin{align}
\mat{G}^{(n)} = \left[
\begin{array}{c}
\tilde{\mat{G}}^{(n)}
\\ \hdashline
\mat{G}_\text{bin}^{(n)}
\end{array}
\right]
\left.
\begin{array}{l}
\big\} \tilde k^{(n)} \times n
\\
\big\} k_\text{bin}^{(n)} \times n
\end{array}
\right\} k^{(n)} \times n
\end{align}
That is, $\mat{G}^{(n)}$ and $\mat{G}_\text{bin}^{(n)}$ are the generating matrices of the fine and coarse codes, respectively.

We can now construct good nested codes in  the following way. We start with random nested codes of rate $(R,\Rbin)$. We interpret the random matrix $\mat{G}^{(n)}$ as consisting of matrices $\tilde G^{(n)}$ and $\mat{G}_\text{bin}^{(n)}$ as above, both i.i.d. Bernoulli-$1/2$. We now add $\mat{G}_I^{(n)} = \mat{G}_I^{(n)}(\mat{G}^{(n)})$ as in the procedure of~\cite{LitsynCoveringCodesBook1997}, to obtain the following generating matrix:
\begin{align}
\mat{G}'^{(n)} = \left[
\begin{array}{c}
\mat{G}_I^{(n)}
\\ \hdashline
\tilde{\mat{G}}^{(n)}
\\ \hdashline
\mat{G}_\text{bin}^{(n)}
\end{array}
\right]
\left.
\begin{array}{l}
\big\} k_I^{(n)} \times n
\\
\big\} \tilde k ^{(n)} \times n
\\
\big\} k_\text{bin}^{(n)} \times n
\end{array}
\right\} k'^{(n)} \times n
\end{align}
Now we construct the fine and coarse codes using the matrices $\mat{G}'^{(n)}$ and $\mat{G}_\text{bin}^{(n)}$, respectively; the rate pair does not change  due to the added matrices. By construction, the fine and coarse codes satisfy \eqref{eq:complicated_covering} and \eqref{eq:complicated_spectrum}, respectively. Thus, by the union bound, the covering property is satisfied with $\delta_C^{(n)}$ and the spectrum property with $\delta_S^{(n)}$, simultaneously, with probability $1-\epsilon_C^{(n)}-\epsilon_S^{(n)}$. We can thus construct a sequence of good nested codes as desired.

\bibliographystyle{IEEEtran}
\bibliography{elih}

\end{document}